\documentclass[letterpaper,11pt]{article}

\usepackage[linesnumbered,vlined]{algorithm2e}
\usepackage[margin=1in]{geometry}
\usepackage[utf8]{inputenc}
\usepackage[dvipsnames]{xcolor}
\usepackage{graphicx}
\usepackage{amsmath,amsfonts,amssymb,amsthm}
\usepackage{lineno}
\usepackage{comment}
\usepackage{enumitem}
\usepackage{hyperref}
\usepackage[capitalise]{cleveref}
\usepackage{etoolbox}
\usepackage{authblk}
\usepackage{thm-restate}
\usepackage{subcaption}

\newtheorem{theorem}{Theorem}
\newtheorem{corollary}[theorem]{Corollary}
\newtheorem{lemma}[theorem]{Lemma}
\newtheorem{definition}[theorem]{Definition}
\newtheorem{observation}[theorem]{Observation}
\newtheorem{remark}[theorem]{Remark}

\crefname{transf}{transformation}{transformations}
\Crefname{transf}{Transformation}{Transformations}

\newtheorem{transf}{Transformation}

\usepackage{tikz}
\usepackage[normalem]{ulem}
\usetikzlibrary{positioning}

\newcommand{\xintprop}{X-intersection Property\xspace}
\newcommand{\yintprop}{Y-intersection Property\xspace}
\newcommand{\extprop}{Extension Property\xspace}

\newcommand{\alex}[1]{\textcolor{red}{#1}} 
\newcommand{\additionalresult}[1]{}

\makeatletter
\let\save@mathaccent\mathaccent
\newcommand*\if@single[3]{%
  \setbox0\hbox{${\mathaccent"0362{#1}}^H$}%
  \setbox2\hbox{${\mathaccent"0362{\kern0pt#1}}^H$}%
  \ifdim\ht0=\ht2 #3\else #2\fi
  }
\newcommand*\rel@kern[1]{\kern#1\dimexpr\macc@kerna}
\newcommand*\widebar[1]{\@ifnextchar^{{\wide@bar{#1}{0}}}{\wide@bar{#1}{1}}}
\newcommand*\wide@bar[2]{\if@single{#1}{\wide@bar@{#1}{#2}{1}}{\wide@bar@{#1}{#2}{2}}}
\newcommand*\wide@bar@[3]{%
  \begingroup
  \def\mathaccent##1##2{%
    \let\mathaccent\save@mathaccent
    \if#32 \let\macc@nucleus\first@char \fi
    \setbox\z@\hbox{$\macc@style{\macc@nucleus}_{}$}%
    \setbox\tw@\hbox{$\macc@style{\macc@nucleus}{}_{}$}%
    \dimen@\wd\tw@
    \advance\dimen@-\wd\z@
    \divide\dimen@ 3
    \@tempdima\wd\tw@
    \advance\@tempdima-\scriptspace
    \divide\@tempdima 10
    \advance\dimen@-\@tempdima
    \ifdim\dimen@>\z@ \dimen@0pt\fi
    \rel@kern{0.6}\kern-\dimen@
    \if#31
      \overline{\rel@kern{-0.6}\kern\dimen@\macc@nucleus\rel@kern{0.4}\kern\dimen@}%
      \advance\dimen@0.4\dimexpr\macc@kerna
      \let\final@kern#2%
      \ifdim\dimen@<\z@ \let\final@kern1\fi
      \if\final@kern1 \kern-\dimen@\fi
    \else
      \overline{\rel@kern{-0.6}\kern\dimen@#1}%
    \fi
  }%
  \macc@depth\@ne
  \let\math@bgroup\@empty \let\math@egroup\macc@set@skewchar
  \mathsurround\z@ \frozen@everymath{\mathgroup\macc@group\relax}%
  \macc@set@skewchar\relax
  \let\mathaccentV\macc@nested@a
  \if#31
    \macc@nested@a\relax111{#1}%
  \else
    \def\gobble@till@marker##1\endmarker{}%
    \futurelet\first@char\gobble@till@marker#1\endmarker
    \ifcat\noexpand\first@char A\else
      \def\first@char{}%
    \fi
    \macc@nested@a\relax111{\first@char}%
  \fi
  \endgroup
}
\makeatother

\newcommand{\macronode}[1]{\mathcal{M}_{#1}}

\SetKwInOut{Input}{Input }
\SetKwInOut{Returns}{Returns }
\SetKwInOut{Outputs}{Outputs }

\SetKwInput{Assert}{Assert}
\SetKw{Let}{let}
\SetKw{Nil}{nil}
\SetKw{Not}{not}
\SetKw{And}{and}
\SetKw{Or}{or}
\SetKw{True}{True}
\SetKw{False}{False}

\SetKwFor{While}{while}{do}{end}
\SetKwRepeat{Until}{repeat}{until}
\SetKwFor{Function}{Function}{}{}
\SetKwFor{RecFun}{Recursive Function}{}{}

\SetCommentSty{}
\SetKwComment{Comment}{$\triangleright$}{}
\SetKw{Return}{return}
\SetKw{Print}{print}

\SetKwProg{Procedure}{Procedure}{}{}

\DontPrintSemicolon

\RestyleAlgo{ruled}

\title{Genome assembly, from practice to theory:\\ safe, complete and \emph{linear-time}}

\author[1]{Massimo Cairo}
\author[2]{Romeo Rizzi}
\author[1]{Alexandru~I.~Tomescu}
\author[3,2]{Elia~C.~Zirondelli}

\affil[1]{Department of Computer Science, University of Helsinki, Finland, \texttt{alexandru.tomescu@helsinki.fi}}
\affil[2]{Department of Computer Science, University of Verona, Italy, \texttt{romeo.rizzi@univr.it}}
\affil[3]{Department of Mathematics, University of Trento, Italy, \texttt{eliacarlo.zirondelli@unitn.it}}

\date{}
\begin{document}

\maketitle

\begin{abstract}
Genome assembly asks to reconstruct an unknown string from many shorter substrings of it. Even though it is one of the key problems in Bioinformatics,
it is generally lacking major theoretical advances. Its hardness stems both from practical issues (size and errors of real data), and from the fact that problem formulations inherently admit multiple solutions. Given these, at their core, most state-of-the-art assemblers are based on finding non-branching paths (\emph{unitigs}) in an assembly graph. While such paths constitute only partial assemblies, they are likely to be correct. More precisely, if one defines a genome assembly solution as a \emph{closed arc-covering walk} of the graph, then unitigs appear in all solutions, being thus \emph{safe} partial solutions.

Until recently, it was open what are \emph{all} the safe walks of an assembly graph. Tomescu and Medvedev (RECOMB 2016) characterized all such safe walks (\emph{omnitigs}), thus giving the first safe and \emph{complete} genome assembly algorithm. Even though omnitig finding was later improved to quadratic time by Cairo et al.~(ACM Trans.~Algorithms 2019), it remained open whether the crucial linear-time feature of finding unitigs can be attained with omnitigs. That is, whether \emph{all} the strings that can be correctly assembled from a graph can be obtained in a time feasible for being implemented in practical genome assemblers.

We answer this question affirmatively, by describing a surprising $O(m)$-time algorithm to \emph{identify} all maximal omnitigs of a graph with $n$ nodes and $m$ arcs, notwithstanding the existence of families of graphs with $\Theta(mn)$ total maximal omnitig size. This is based on the discovery of a family of walks (\emph{macrotigs}) with the property that all the non-trivial omnitigs are univocal extensions of subwalks of a macrotig.
This has two consequences: 
\begin{enumerate}
    \item A \emph{linear-time output-sensitive} algorithm enumerating all maximal omnitigs.
    \item A \emph{compact $O(m)$ representation} of all maximal omnitigs, which allows, e.g., for $O(m)$-time computation of various statistics on them.
\end{enumerate}
Our results close a long-standing theoretical question inspired by practical genome assemblers, originating with the use of unitigs in 1995. They are also crucial in covering problems incorporating additional practical constraints. Thus, we envision our results to be at the core of a reverse transfer from theory to practical and \emph{complete} genome assembly programs, as has been the case for other key Bioinformatics problems.
\end{abstract}

\thispagestyle{empty}

\newpage
\clearpage
\setcounter{page}{1}

\section{Introduction}

\paragraph{Theoretical and practical background of genome assembly.} Genome assembly is one of the flagship problems in Bioinformatics, along with other problems originating in--or highly motivated by--this field, such as edit distance computation, reconstructing and comparing phylogenetic trees, text indexing and compression. In genome assembly, we are given a collection of strings (or \emph{reads}) and we need to reconstruct the unknown string (the genome) from which they originate. This is motivated by sequencing technologies that are able to read either ``short'' strings (100-250 length, Illumina technology), or ``long'' strings (10.000-50.000 length, Pacific Biosciences or Oxford Nanopore technologies) in huge amounts from the genomic sequence(s) in a sample. For example, the SARS-CoV-2 genome was obtained in~\cite{Wu:2020aa} from short reads using the MEGAHIT assembler~\cite{li2015megahit}. 

Other leading Bioinformatics problems have seen significant theoretical progress in major Computer Science venues, culminating (just to name a few) with both positive results, see e.g.~\cite{10.1145/3313276.3316390,WilliamsWWY15} for phylogeny problems, \cite{DBLP:conf/focs/FerraginaM00,DBLP:conf/stoc/Belazzougui14,DBLP:conf/stoc/KempaK19} for text indexing, \cite{DBLP:conf/soda/FerraginaNV09,DBLP:conf/soda/BelazzouguiP16,DBLP:conf/stoc/KempaP18} for text compression, and negative results, see e.g.~\cite{DBLP:conf/stoc/BackursI15,DBLP:conf/focs/AbboudBW15,DBLP:conf/focs/BackursI16,DBLP:conf/icalp/EquiGMT19} for string matching problems. However, the genome assembly problem is generally lacking major theoretical advances. 

One reason for this stems from practice: the huge amount of data (e.g.~the 3.1 Billion long human genome is read 50 times over) which impedes slower than linear-time algorithms, errors of the sequencing technologies (up to 15\% for long reads), and various biases when reading certain genomic regions~\cite{nagarajan2013sequence}. 
Another reason stems from theory: historically, finding an optimal genome assembly solution is considered NP-hard under several formulations~\cite{peltola83,K92,KM95,MB09,nagarajan2009parametric,SequencingHybridizationLysov1988,narzisi14}, but, more fundamentally, even if one outputs a 3.1 Billion characters long string, this is likely incorrect, since problem formulations inherently admit a large number of solutions of such length~\cite{kingsford10}.

Given all these setbacks, most state-of-the-art assemblers, including e.g.~MEGAHIT~\cite{li2015megahit} (for short reads), or wtdbg2~\cite{Ruan:2020aa} (for long reads), generally employ a very simple and \emph{linear-time} strategy, dating back to 1995~\cite{KM95}. They start by building an assembly graph encoding the overlaps of the reads, such as a \emph{de Bruijn graph} \cite{Pevzner1989} or an \emph{overlap graph}~\cite{M05} (graphs are directed in this paper). After some simplifications to this graph to remove practical artifacts such as errors, at their core they find strings labeling paths whose internal nodes have in-degree and out-degree equal to 1 (called \emph{unitigs}), approach dating back to 1995~\cite{KM95}. That is, they do not output \emph{entire} genome assemblies, but only shorter strings that are likely to be present in the sequenced genome, since unitigs do not branch at internal nodes.

\smallskip


The issue of multiple solutions to a problem has deep roots in Bioinformatics, but is in fact common to many real-world problems from other fields. In such problems, we seek ultimate knowledge of an unknown object (e.g., a genome) but have access only to partial observations from it (e.g., reads). A standard paradigm is to apply Occam's razor principle which favors the simplest model explaining the data. As such, the reconstruction problem is cast in terms of an optimization problem, to be addressed by various mathematical, computational and technological paradigms. 

While this approach has been extremely successful in Bioinformatics, it is not always robust. First, the optimization problem might admit several optimal solutions, and thus several interpretations of the observed data. A standard way to tackle this is to \emph{enumerate} all solutions~\cite{Grossi2016,Kiyomi2016}. Second, the problem formulation might be inaccurate, or the data might be incomplete, and the true solution might be a sub-optimal one. One could then enumerate all the \emph{first $k$-best} solutions to it~\cite{DBLP:journals/eatcs/Eppstein15,DBLP:reference/algo/Eppstein15}, hoping that the true solution is among such first $k$ ones. The motivation of such enumeration algorithms is that e.g.~later ``one can apply more sophisticated quality criteria, wait for data to become available to choose among them, or present them all to human decision-makers''~\cite{DBLP:journals/eatcs/Eppstein15}. However, both approaches do not scale when the number of solutions is large, and are thus unfeasible in genome assembly.

\paragraph{Safe and complete algorithms: A theoretical framing of practical genome assembly.} With the aim of enhancing the widely-used practical approach of assembling just unitigs---as those walks considered to be present in any possible assembly solution---a result in a major Bioinformatics venue~\cite{tomescu2017safe} asked \emph{what is the ``limit'' of the correctly reconstructible information from an assembly graph}. Moreover, is all such reconstructible information still obtainable in \emph{linear time}, as in the case of the popular unitigs? Variants of this question also appeared in~\cite{journals/bioinformatics/Guenoche92,boisvert2010ray,nagarajan2009parametric,DBLP:journals/bioinformatics/ShomoronyKCT16,DBLP:journals/bmcbi/LamKT14,BBT13}, while other works already considered simple linear-time generalizations of unitigs~\cite{PTW01,MGMB07,jacksonthesis,kingsford10}, without knowing if the ``assembly limit'' is reached.

To make this question precise,~\cite{tomescu2017safe} introduced the following \emph{safe and complete} framework. Given a notion of solution to a problem (e.g.~a type of walk in a graph), a partial solution (e.g.~some shorter walk in the graph) is called \emph{safe} if it appears (e.g.~is a subwalk) in all solutions. An algorithm reporting only safe partial solutions is called a \emph{safe algorithm}. A safe algorithm reporting \emph{all} safe partial solutions is called \emph{safe and complete}. A safe and complete algorithm outputs all and only what is likely part of the unknown object to be reconstructed, \emph{synthesizing all solutions from the point of view of correctness.}

Safety generalizes the existing notion of \emph{persistency}: a \emph{single} node or edge was called \emph{persistent} if it appears in all solutions~\cite{doi:10.1137/0603052,Costa1994143,DBLP:journals/mmor/Cechlarova98}, for example persistent edges for maximum bipartite matchings~\cite{Costa1994143}. However, it also has roots in other Bioinformatics works~\cite{Vingron01071990,Chao01081993,Friemann:1992aa,Zuker:1991aa}, which considered the aligned symbols---the \emph{reliable regions}---appearing in all optimal (and sub-optimal) alignments of two strings. The reliable regions of an alignment of proteins were shown in \cite{Vingron01071990} to match in a significant proportion the true ones determined experimentally.

There are many theoretical formulations of genome assembly as an optimization problem, e.g.~a shortest common superstring of all the reads~\cite{peltola83,K92,KM95}, or some type of shortest walk covering all nodes or arcs of the assembly graph~\cite{PTW01,MB09,MGMB07,kapun13b,SequencingHybridizationLysov1988,narzisi14,nagarajan2009parametric}. However, it is widely acknowledged~\cite{nagarajan2009parametric,narzisi14,medvedev2019modeling,nagarajan2013sequence,DBLP:books/cu/MBCT2015,kingsford2010assembly} that, apart from some being NP-hard, these formulations are lacking in several aspects, for example they collapse repeated regions of a genome. At present, given the complexity of the problem, there is no definitive notion of a ``good'' genome assembly solution. Therefore,~\cite{tomescu2017safe} considered as genome assembly solution \emph{any} closed arc-covering walk of a graph, where \emph{arc-covering} means that it passes through each arc \emph{at least} once. The main benefit of considering \emph{any} arc-covering walk is that safe walks for them are safe also for any possible restriction such covering walks (e.g.~by some additional optimality criterion\footnote{For example, closed arc-covering walks are a common relaxation of the fundamental notions of closed Eulerian walk (we now pass through each arc \emph{at least once}), and of closed Chinese postman walk (i.e.~a closed arc-covering walk \emph{of minimum length})~\cite{guan1962graphic}, which were mentioned in \cite{nagarajan2009parametric} as unsatisfactory models of genome assembly.}). Put otherwise, safe walks for \emph{all} arc-covering walks are more likely to be correct than safe walks for some peculiar type of arc-covering walks.

\begin{figure}[t!]
	\centering
	\vspace{-.5cm}
	\includegraphics[scale=0.20]{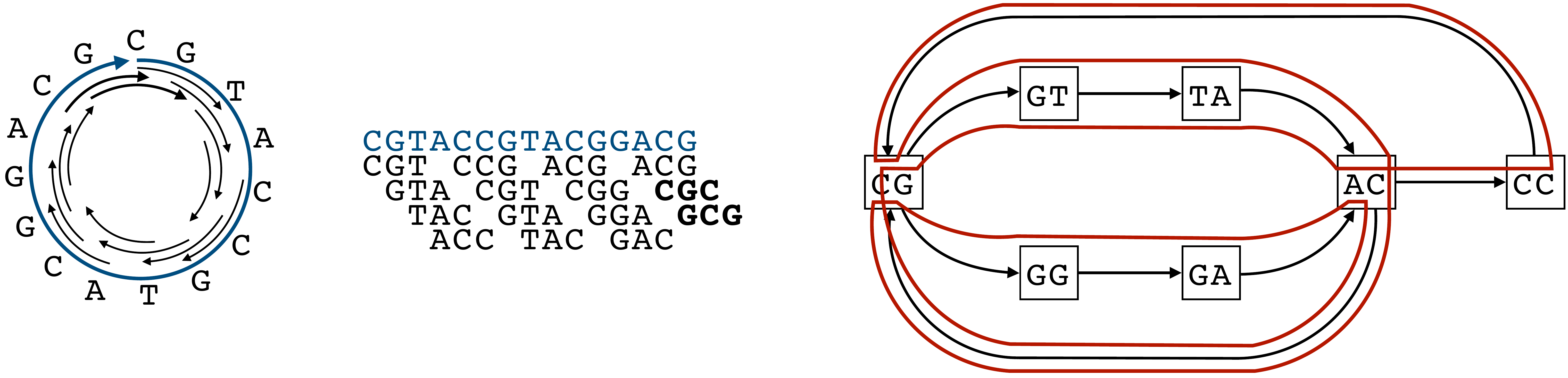}
	\caption{\small
		Left: A circular string and a set strings of length 3, drawn as black arrows. Middle: The same circular string shown as linearized (in blue) and the same strings (in black). The two strings in bold overlap the beginning and the end of the circular string. Right: The de Bruijn graph of order $k=3$ built from the input strings. The set of input string corresponds to the set of arcs. The closed arc-covering walk spelling the circular string on the left is shown in red. However, the graph admits several closed arc-covering walks (also an Eulerian one), all spelling thus a circular string with the same set of $k$-mers as the circular string on the left.	\label{fig:db-example-0}}
\vspace{-.3cm}
\end{figure}

Moreover, closed arc-covering walks in widely-used de Bruijn graphs are in one-to-one correspondence with circular strings having the same set of $k$-mers of the reads (a \emph{$k$-mer} of a set of strings is any length-$k$ string appearing as a substring of some string in the set). More precisely, consider the following setting mentioned in~\cite{tomescu2017safe}. A \emph{de Bruijn graph of order $k$} has the set of all \emph{$(k-1)$-mers} of the reads as the set of nodes, and the set of all $k$-mers of the reads as arcs (from their length-$(k-1)$ prefix to their length-$(k-1)$ suffix). The most basic notion of genome assembly solution (for circular genomes) are circular strings having the same set of $k$-mers as the reads, which correspond exactly to closed arc-covering walks of the de Bruijn graph of order $k$. See \Cref{fig:db-example-0,sec:bioinfo-motivation} for a more detailed description.

\paragraph{Prior results on safety in closed arc-covering walks.} It is immediate to see that unitigs are safe walks for closed arc-covering walks. A first safe generalization of unitigs consisted of those paths whose internal nodes have only out-degree equal to 1 (with no restriction on their in-degree)~\cite{PTW01}. Further, these safe paths have been generalized in~\cite{MGMB07,jacksonthesis,kingsford10} to those partitionable into a prefix whose nodes have in-degree equal to 1, and a suffix whose nodes have out-degree equal to 1. \emph{All} safe walks for closed arc-covering walks were characterized by~\cite{tomescu2017safe,TomescuMedvedev} with the notion of \emph{omnitigs}, see \Cref{def:omnitig,fig:omnitig}. This lead to the first \emph{safe and complete} genome assembly algorithm (obtained thus 20 years after unitigs were first considered), outputting all \emph{maximal} omnitigs in polynomial time (maximal omnitigs are those which are not sub-walks of other omnitigs).

\begin{definition}[Omnitig]
\label{def:omnitig}
A walk $W = e_0\dots e_\ell$ is an \emph{omnitig} if, for all $1 \leq i \leq j \leq \ell$, there is no non-empty path (\emph{forbidden path}) from the tail of $e_j$ to the head of $e_{i-1}$, with first arc different from $e_j$ and last arc different from $e_{i-1}$.
\end{definition}

\begin{figure}[h!]
	\centering
	\includegraphics{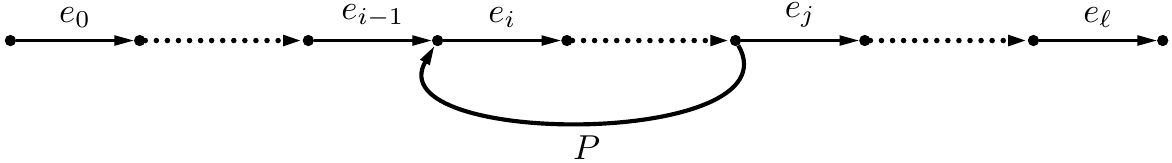}
	\caption{\small
		The walk $e_0\dots e_\ell$ is \emph{not} an omnitig because there is a non-empty path forbidden path $P$.
		}
	\label{fig:omnitig}
\end{figure}

Furthermore, through experiments on ``perfect'' human read datasets,~\cite{tomescu2017safe} also showed that strings labeling omnitigs are about 60\% longer on average than unitigs, and contain about 60\% more biological content on average. Thus, once other issues of real data (e.g.~errors) are added to the problem formulation, omnitigs (and the safe walks for such extended models) have the potential to significantly improve the quality of genome assembly results. Nevertheless, for this to be possible, one first needs the best possible results for omnitigs (given e.g.~the sheer size of the read datasets), and a full comprehension of them, otherwise, such extensions are hard to solve efficiently.

Cairo et al.~\cite{DBLP:journals/talg/CairoMART19} recently proved that the length of all maximal omnitigs of any graph with $n$ nodes and $m$ arcs is $O(nm)$, and proposed an $O(nm)$-time algorithm enumerating all maximal omnitigs. This was also proven to be optimal, in the sense that they constructed families of graphs where the total length of all maximal omnitigs is $\Theta(nm)$. However, it was left open if it is necessary to pay $O(nm)$ even when the total length of the output is smaller. Moreover, that algorithm cannot break this barrier, because e.g.~$O(m)$-time traversals have to be done for $O(n)$ cases.

This theoretical question is crucial also from the practical point of view: assembly graphs have the number of nodes and arcs in the order of millions, and yet the total length of the maximal omnitigs is almost linear in the size of the graph. For example, the compressed (see \Cref{def:compressed-graph}) de Bruijn graph of human chromosome 10 (length 135 million) has 467 thousand arcs~\cite[Table~1]{DBLP:journals/talg/CairoMART19}, and the length of all maximal omnitigs (i.e.~their total number of arcs, not their total string length) is 893 thousand. Moreover, even though this chromosome is only about 4\% of the full human genome, the running time of the \emph{quadratic} algorithm of~\cite{DBLP:journals/talg/CairoMART19} on its compressed de Bruijn graph is about 30 minutes. Both of these facts imply that a linear-time output sensitive enumeration algorithm has also a big potential for practical impact.

\paragraph{Our results.} Our main result is an $O(m)$-size representation of all maximal omnitigs\footnote{Note that the total length of the maximal omnitigs is at least $m$, since every arc is an omnitig.}, based on a careful structural decomposition of the omnitigs of a graph. This is surprising, given that there are families of graphs with $\Theta(nm)$ total length of maximal omnitigs~\cite{DBLP:journals/talg/CairoMART19}.

\begin{theorem}
\label{thm:macrotigs-main}
Given a strongly connected graph $G$ with $n$ nodes and $m$ arcs, there exists a $O(m)$-size representation of all maximal omnitigs, consisting of a set $\mathcal{M}$ of walks (\emph{maximal macrotigs}) of total length $O(n)$ and a set $\mathcal{F}$ of arcs, such that every maximal omnitig is the univocal extension\footnote{The \emph{univocal extension} $U(W)$ of a walk $W$ is obtained by appending to $W$ the longest path whose nodes (except for the last one) have out-degree $1$, and prepending to $W$ the longest path whose nodes (except for the first one) have in-degree $1$; see~\Cref{sec:overview} for the formal definition.} of either a subwalk of a walk in $\mathcal{M}$, or of an arc in $\mathcal{F}$. 

Moreover, $\mathcal{M}$, $\mathcal{F}$, and the endpoints of macrotig subwalks univocally extending to maximal omnitigs can be computed in time $O(m)$.
\end{theorem}

Since the univocal extension $U(W)$ of a walk $W$ can be trivially computed in time linear in the length of $U(W)$, we get the linear-time output sensitive algorithm as an immediate corollary:

\begin{corollary}
\label{thm:main-1}
Given a strongly connected graph $G$, it is possible to enumerate all maximal omnitigs of $G$ in time linear in their total length. 
\end{corollary}

We obtain \Cref{thm:macrotigs-main} using two interesting ingredients. The first is a novel graph structure (\emph{macronodes}), obtained after a \emph{compression} operation of `easy' nodes and arcs (\Cref{sec:macronode-macrotigs}). The second is a connection to a recent result by Georgiadis et al.~\cite{georgiadis2017strong} showing that it is possible to answer in $O(1)$-time strong connectivity queries under a \emph{single} arc removal, after a linear-time preprocessing of the graph (notice that a forbidden path is defined w.r.t.~\emph{two} arcs to avoid). 

\Cref{thm:macrotigs-main} has additional practical implications. First, omnitigs are also representable in the same (linear) size as the commonly used unitigs. Second, maximal macrotigs enable various $O(m)$-time operations on maximal omnitigs (without listing them explicitly), by pre-computing the univocal extensions from any node, needed in \Cref{thm:macrotigs-main}. For example, given that the number of maximal omnitigs is $O(m)$~\cite{DBLP:journals/talg/CairoMART19}, this implies the following result:

\begin{corollary}
\label{cor:enumerate-lengths}
Given a strongly connected graph $G$ with $m$ arcs, it is possible to compute the lengths of all maximal omnitig in total time $O(m)$.
\end{corollary}

\Cref{cor:enumerate-lengths} leads to a linear-time computation of various statistics about maximal omnitigs, such as minimum, maximum, and average length (useful e.g.~in~\cite{10.1093/bioinformatics/btt310}). One can also use this to filter out subfamilies of them (e.g.~those of length smaller and / or larger than a given value) before enumerating them explicitly.

\paragraph{Significance of our results.} This paper closes the issue of finding safe walks for a fundamental model of genome assembly (\emph{any} closed arc-covering walk), a long-standing theoretical question, inspired by practical genome assemblers, and originating with the use of unitigs in 1995~\cite{KM95}. However, we envision a reverse transfer from theory to practical and \emph{complete} genome assembly programs, as has been the case in other Bioinformatics problems. 

Trivially, safe walks for all closed arc-covering walks are also safe for more specific types of arc-covering walks. Moreover, while a genome soltuion defined as a single closed arc-covering walk does not incorporate several practical issues of real data, in a follow-up work~\cite{hydrostructure} we show that omnitigs are the basis of more advanced models handling many practical aspects. For example, to allow more types of genomes to be assembled, one can define an assembly solution as a \emph{set} of closed walks that together cover all arcs~\cite{DBLP:journals/almob/AcostaMT18}, which is the case in \emph{metagenomic} sequencing of bacteria. For linear chromosomes (as in eukaryotes such as human), or when modeling missing sequencing coverage, one can analogously consider one, or many, such \emph{open} walks~\cite{TomescuMedvedev,tomescu2017safe}. Safe walks for all these models are subsets of omnitigs~\cite{DBLP:journals/almob/AcostaMT18,hydrostructure}. Moreover, when modeling sequencing errors, or mutations present e.g.~only in the mother copy of a chromosome (and not in the father's copy), one can require some arcs not to be covered by a solution walk, or even to be ``invisible'' from the point of view safety. Finding safe walks for such models is also based on first finding omnitigs-like walks~\cite{hydrostructure}.

Notice that such separation between theoretical formulations and their practical embodiments is common for many classical problems in Bioinformatics. For example, computing edit distance is often replaced with computing edit distance under affine gap costs~\cite{durbin1998biological}, or enhanced with various heuristics as in the well-known BLAST aligner~\cite{altschul1990basic}. Also text indexes such as the FM-index~\cite{DBLP:conf/focs/FerraginaM00} are extended in popular read mapping tools (e.g.~\cite{li2009fast,langmead2012fast}) with many heuristics handling errors and mutations in the reads.

Finally, our results show that safe partial solutions enjoy interesting combinatorial properties, further promoting the persistency and safety frameworks. For real-world problems admitting multiple solutions, safe and complete algorithms are more pragmatic than the classical approach of outputting an arbitrary optimal solution. They are also more efficient than enumerating all solutions, or only the first $k$-best solutions, because they already \emph{synthesize all that can be correctly reconstructed from the input data.}

\section{Overview of the proofs}
\label{sec:overview}

We highlight here our key structural and algorithmic contributions, and give the formal details in \Cref{sec:macronode-macrotigs,sec:final-algorithm,sec:prepro}. We start with the minimum terminology needed to understand this section, and defer the rest of the notation to \Cref{sec:preliminaries}.

\noindent \paragraph*{Terminology.} Functions $t(\cdot)$ and $h(\cdot)$ denote the \emph{tail} node and the \emph{head} node, respectively, of an arc or walk. We classify the nodes and arcs of a strongly connected graph as follows (see \Cref{fig:macronodes-example}):
\begin{itemize}
	\itemsep0em
	
	\item a node $v$ is a \emph{join node} if its \emph{in-degree} $d^-(v)$ satisfies $d^-(v) > 1$, and a \emph{join-free node} otherwise. An arc $f$ is called a \emph{join arc} if $h(f)$ is a join node, and a \emph{join-free arc} otherwise. 
	\item a node $v$ is a \emph{split node} if its \emph{out-degree} $d^+(v)$ satisfies $d^+(v) > 1$, and a \emph{split-free node} otherwise. An arc $g$ is called a \emph{split arc} if $t(g)$ is a split node, and a \emph{split-free arc} otherwise. 
	\item a node or arc is called \emph{bivalent} if it is both join and split, and it is called \emph{biunivocal} if it is both split-free and join-free.
\end{itemize}
A walk $W$ is \emph{split-free} (resp., \emph{join-free}) if all its arcs are split-free (resp., join-free). Given a walk $W$, its \emph{univocal extension $U(W)$}\label{def:univocal-extension} is defined as $W^-WW^+$, where $W^-$ is the longest join-free path to $t(W)$ and $W^+$ is the longest split-free path from $h(W)$ (observe that they are uniquely defined).

\noindent \paragraph*{Structure.} The main structural insight of this paper is that omnitigs enjoy surprisingly limited freedom, in the sense that any omnitig can be seen as a concatenation of walks in a very specific set. In order to give the simplest exposition, we first simplify the graph by contracting biunivocal nodes and arcs. The nodes of the resulting graph can now be partitioned into \emph{macronodes} (see \Cref{fig:macronodes-example,def:macronode,}), where each macronode $\macronode{v}$ is uniquely identified by a bivalent node~$v$ (its \emph{center}).
We can now split the problem by first finding omnitigs inside each macronode, and then characterizing the ways in which omnitigs from different macronodes can combine.

\begin{figure}[t]
    \centerline{
    \begin{subfigure}[m]{0.60\linewidth}
	\centering
	\includegraphics[width=\textwidth]{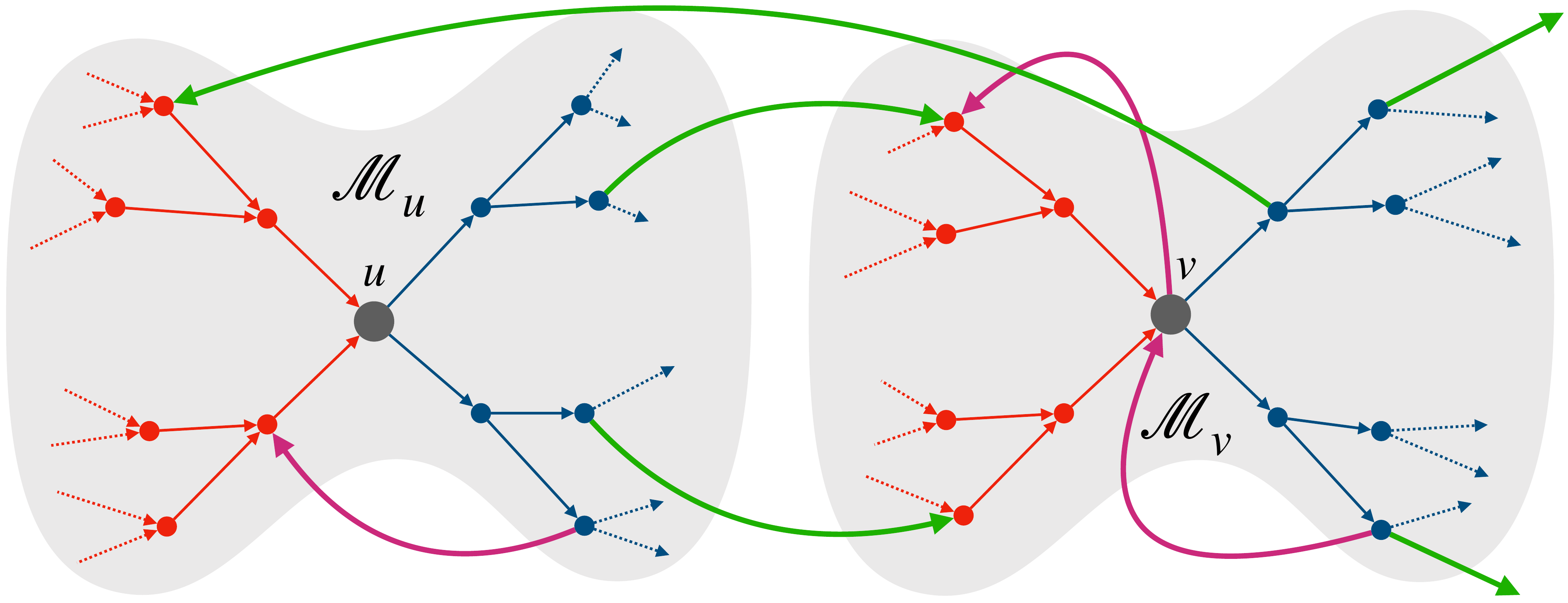}
	\caption{\label{fig:macronodes-example}}
	\end{subfigure}
	\hfill
	\begin{subfigure}[m]{0.35\linewidth}
	\centering
	\includegraphics[width=\textwidth]{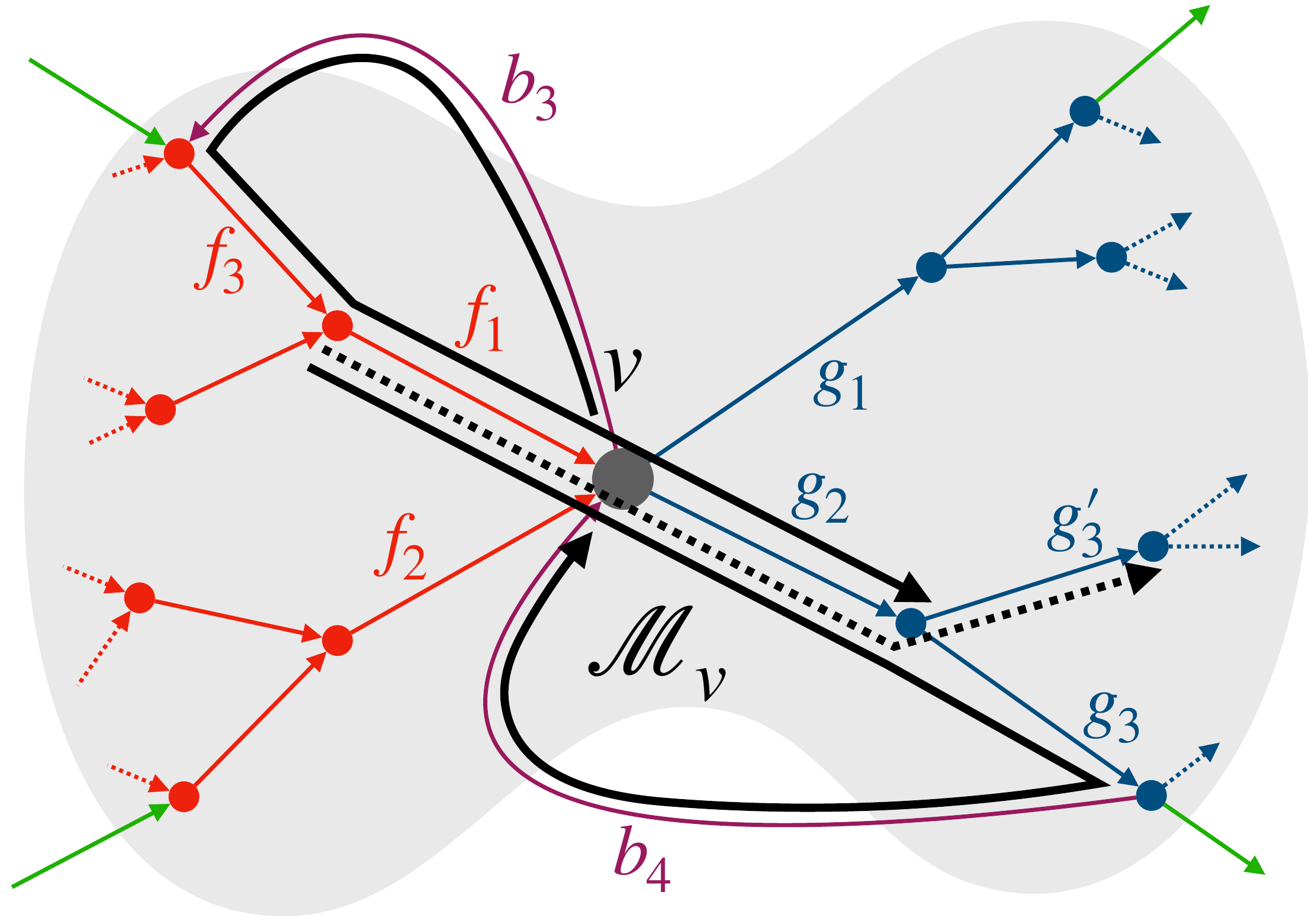}
	\caption{\label{fig:microtigs}}
	\end{subfigure}
	}
	
	\caption{\small Left: Given a bivalent node $v$, the macronode $\mathcal{M}_v$ is the subgraph of $G$ induced by the nodes reaching $v$ with a split-free path (in red), and the nodes reachable from $v$ with a join-free path (in blue). These two types of nodes induce the two trees of the macronode. By definition, every arc with endpoints in different macronodes are bivalent (in green). The remaining bivalent arcs have endpoints in the same macronode (in purple). 
	Right: The only omnitig traversing the bivalent node $v$ is $f_1g_2$; e.g., by the \xintprop neither $f_2g_2$ is an omnitig ($b_3f_3f_1$ is a forbidden path) nor $f_1g_1$ is an omnitig ($g_2g_3b_4$ is a forbidden path). Extending the micro-omnitig $f_1g_2$ to the right we notice that $f_1g_2g_3$ is an omnitig and by the \yintprop $f_1g_2g'_3$ is not an omnitig ($g_3b_4$ is a forbidden path). Hence, the only maximal right-micro omnitig is $f_1g_2g_3b_4$, and the only maximal left-micro omnitig is $b_3f_3f_1g_2$. Merging the two on $f_1g_2$, we obtain the maximal microtig $b_3f_3f_1g_2g_3b_4$.}
\end{figure}


We discover a key combinatorial property of how omnitigs can be extended: there are at most two ways that any omnitig can traverse a macronode center (see also \Cref{fig:microtigs}):

\begin{theorem}[\xintprop]
\label{thm:Xthm}
Let $v$ be a bivalent node. Let $f_1$ and $f_2$ be distinct join arcs with $h(f_1) = h(f_2) = v$; let $g_1$ and $g_2$ be distinct split arcs with $t(g_1) = t(g_2) = v$. We have:
\begin{itemize}
	\itemsep0em
	\item[$i)$] If $f_1g_1$ and $f_2g_2$ are omnitigs, then $d^+(v)=d^-(v)=2$.
	\item[$ii)$] If $f_1g_1$ is an omnitig, then there are no omnitigs $f_1 g'$ with $g' \neq g_1$, nor $f'g_1$ with $f' \neq f_1$. 
\end{itemize}
\end{theorem}

In order to prove the \xintprop, we prove an even more fundamental property: once an omnitig traverses a macronode center, for any node it meets after the center node, there is at most one way of continuing from that node (\yintprop, \Cref{cor:Y-lemma}), see \Cref{fig:microtigs}. The basic intuition is that if there are more than one possibilities, then strong connectivity creates forbidden paths. 

Given an omnitig $fg$ traversing the bivalent node $v$, we define the \emph{maximal right-micro omnitig} as the longest extension $fgW$ in the macronode $\macronode{v}$ (see \Cref{fig:microtigs} and \Cref{def:omni_type}). The \emph{maximal left-micro omnitig} is the symmetrical omnitig $Wfg$. By \Cref{thm:Xthm}, there are at most two maximal right-micro omnitigs and two maximal left-micro omnitigs. The merging of a maximal left- and right-micro omnitig on $fg$ is called a \emph{maximal microtig} (see \Cref{fig:microtigs,def:omni_type}; notice that a microtig is not necessarily an omnitig). These \emph{at most two} maximal microtigs represent ``forced omnitig tracks'' that must be followed by any omnitig crossing $v$.



We now describe how omnitigs can advance from one macronode to another. Notice that any arc having endpoints in different macronodes is a bivalent arc (\Cref{obs:macropartition}). In \Cref{lem:arcs-in-microtigs} we prove that for every maximal microtig ending with a bivalent arc $b$, there is at most one maximal microtig starting with $b$. As such, when an omnitig track exits a macronode, there is at most one way of connecting it with an omnitig track from another macronode. 
It is natural to merge all omnitig tracks (i.e.~maximal microtigs) on all bivalent arcs between \emph{different} macronodes, and thus obtain \emph{maximal macrotigs} (\Cref{def:macrotig,fig:macrotig-example}). The total size of all maximal macrotigs is $O(n)$ (\Cref{lem:alg-extend-correct}), and they are a representation of all maximal omnitigs, except for those that are univocal extensions of the arcs of $\mathcal{F}$, see below and \Cref{lem:omnijoinsplitproto}.

\noindent \paragraph*{Algorithms.} Our algorithms first build the set $\mathcal{M}$ of maximal macrotigs, and then identify maximal omnitigs inside them. The set $\mathcal{F}$ of arcs univocally extending to the remaining maximal omnitigs will be the set of bivalent arcs not appearing in $\mathcal{M}$ (\Cref{lem:omnijoinsplitproto}). 

Crucial to the algorithms is an \emph{extension} primitive deciding what new arc (if any) to choose when extending an omnitig (recall that the X- and Y-intersection Properties limits the \emph{number} of such arcs to one). Suppose we have an omnitig $fW$, with $f$ a join arc, and we need to decide if it can be extended with an arc $g$ out-going from $h(W)$. Naturally, this extension can be found by checking that there is no forbidden path  from $t(g) = h(W)$. However, this forbidden path can potentially end in any node of $W$. Up to this point, \cite{TomescuMedvedev, tomescu2017safe,DBLP:journals/talg/CairoMART19} need to do an entire $O(m)$ graph traversal to check if any node of $W$ is reachable by a forbidden path.
%
%
We prove here a new key property:

\begin{theorem}[\extprop]
\label{thm:strong-ita-result}
Let $fW$ be an omnitig in $G$, where $f$ is a join arc. 
Then $fWg$ is an omnitig if and only if $g$ is the only arc with $t(g) = h(W)$ such that there exists a path from $h(g)$ to $h(f)$ in $G \smallsetminus f$.
\end{theorem}

Thus, for each arc~$g$ with $t(g) = h(W)$, we can do a \emph{single} reachability query under one arc removal: ``does $h(g)$ reach $h(f)$ in $G \smallsetminus f$?'' Since the target of the reachability query is also the head of the arc excluded $f$, then we can apply an immediate consequence of the results of \cite{georgiadis2017strong}:

\begin{theorem}[\cite{georgiadis2017strong}]
\label{cor:italiano_result}
Let $G$ be a strongly connected graph with $n$ nodes and $m$ arcs. It is possible to build an $O(n)$-space data structure that, after $O(m+n)$-time preprocessing, given a node $w$ and an arc $f$, tests in $O(1)$ worst-case time if there is a path from $w$ to $h(f)$ in $G \smallsetminus f$.
\end{theorem}

Using the \extprop and \Cref{cor:italiano_result}, we can thus pay $O(1)$ time to check each out-outgoing arc $g$, before discovering the one (if any) with which to extend $fW$. In \Cref{sec:prepro} we describe how to transform the graph to have constant degree, so that we pay $O(1)$ per node. This transformation also requires slight changes to the maximal omnitig enumeration algorithm to maintain the linear-time output sensitive complexity (see~\Cref{subsec:constant-degree-enumeration}). We first use the \extprop when building the left- and right-maximal micro omnitigs, and then when identifying maximal omnitigs inside macrotigs, as follows.

\begin{figure}[t!]
	\centering
\centerline{	
	\begin{subfigure}[t]{0.62\textwidth}
    \centering
    \includegraphics[width=\textwidth]{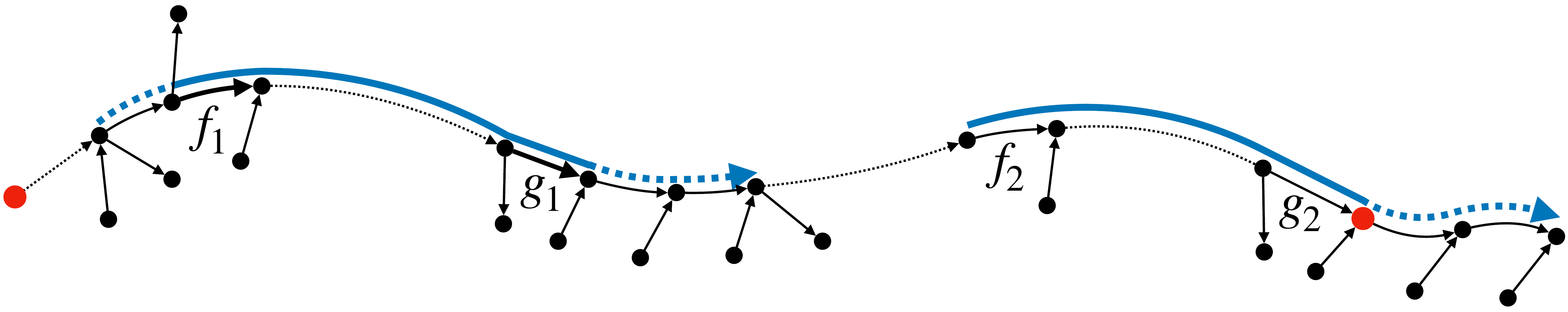}
    \end{subfigure}
    \hfill
    \begin{subfigure}[t]{0.32\textwidth}
    \centering
    \includegraphics[width=\textwidth]{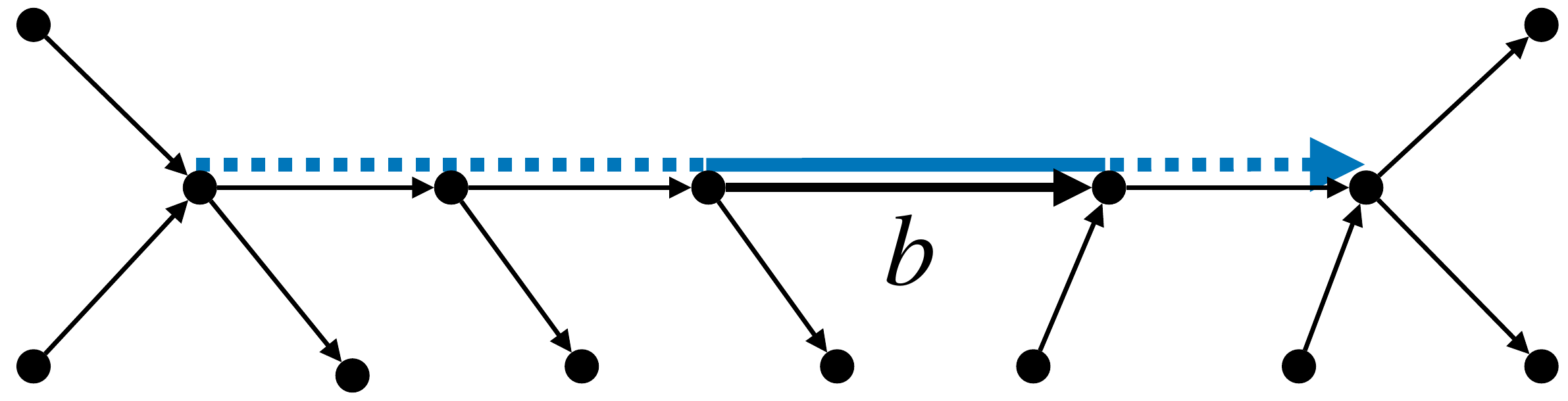}
    \end{subfigure}
}
	\caption{\small Any maximal omnitig is identified (in solid blue) either by a macrotig interval (from a join arc $f$ to a split arc $g$; left), or by a bivalent arc $b$ not appearing in any macrotig (right). The full maximal omnitig is obtained by univocal extension (dotted blue), extension which may go outside of the maximal macrotig. 
	\label{fig:macronode-algorithm}}
\end{figure}

Once we have the set $\mathcal{M}$ of maximal macrotigs, we scan each macrotig with two pointers, a left one always on a join arc $f$, and a right one always on a split arc $g$ (see \Cref{fig:macronode-algorithm,alg:maximalomnitigs}). Both pointers move from left to right in such a way that the subwalk between them is always an omnitig. The subwalk is grown to the right by moving the right pointer as long as it remains an omnitig (checked with the \extprop).
When growing to the right is no longer possible, the omnitig is shrunk from the left by moving the left pointer. This technique runs in time linear to the total length of the maximal macrotigs, namely $O(n)$. 


In \Cref{fig:macronodes} we work out all these notions on a concrete example.

\begin{figure}[h!]
\centerline{
    \begin{subfigure}[t]{0.5\textwidth}
    \centering
    \includegraphics[width=\textwidth]{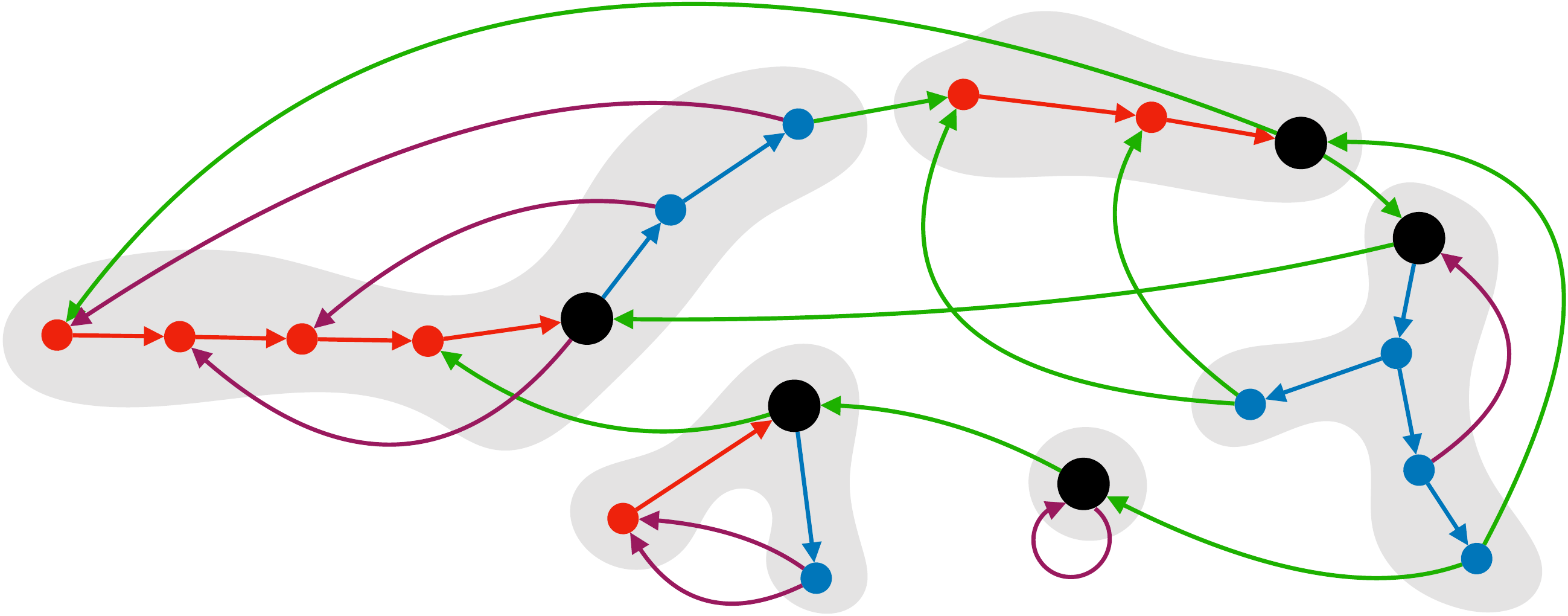}
    \caption{\small Nodes and arcs color-coded as in \Cref{fig:macronodes-example}. \label{fig:macronodes-A}}
    \end{subfigure}
    \begin{subfigure}[t]{0.5\textwidth}
    \centering
    \includegraphics[width=\textwidth]{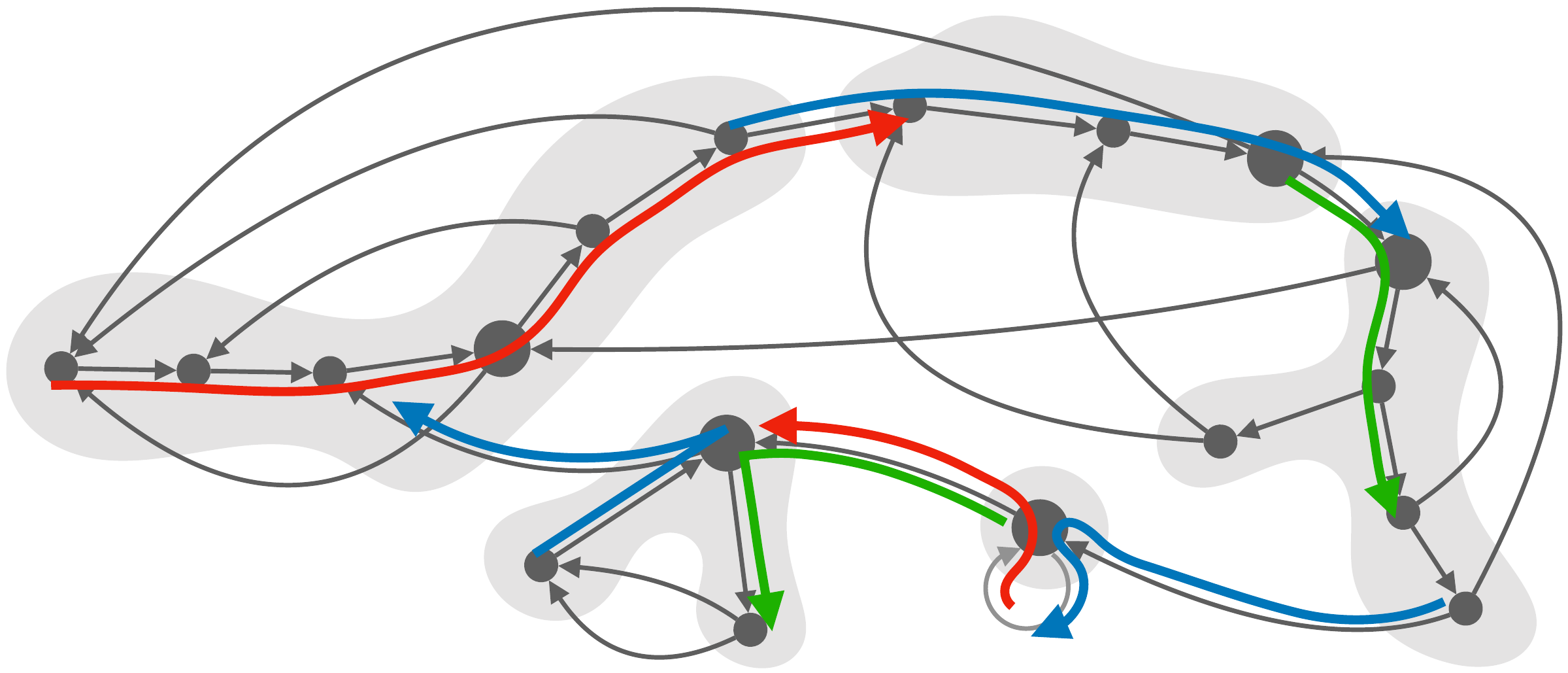}
    \caption{\small Maximal microtigs.\label{fig:macronodes-B}}
    \end{subfigure}
}
\centerline{
    \begin{subfigure}[t]{0.5\textwidth}
    \centering
    \includegraphics[width=\textwidth]{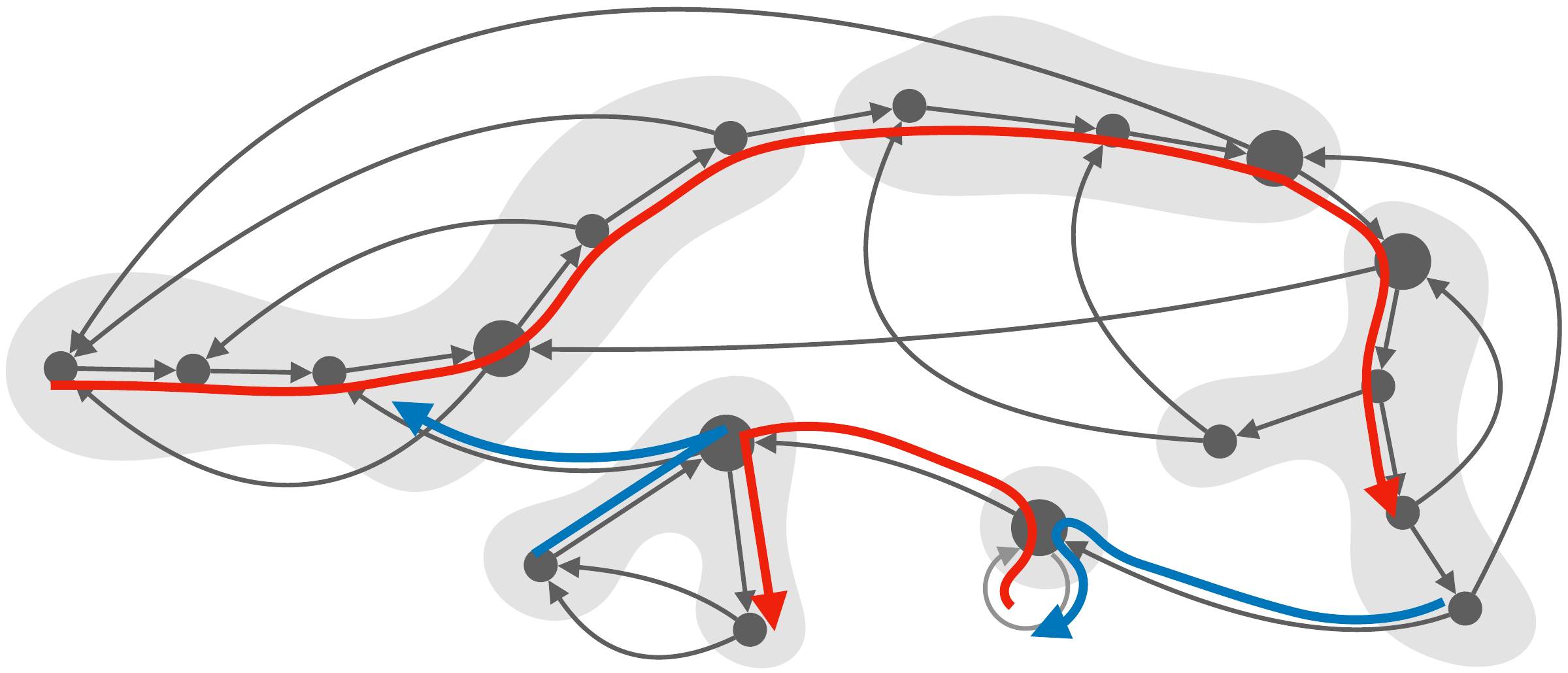}
    \caption{\small Maximal macrotigs.\label{fig:macronodes-C}}
    \end{subfigure}
    \begin{subfigure}[t]{0.5\textwidth}
    \centering
    \includegraphics[width=\textwidth]{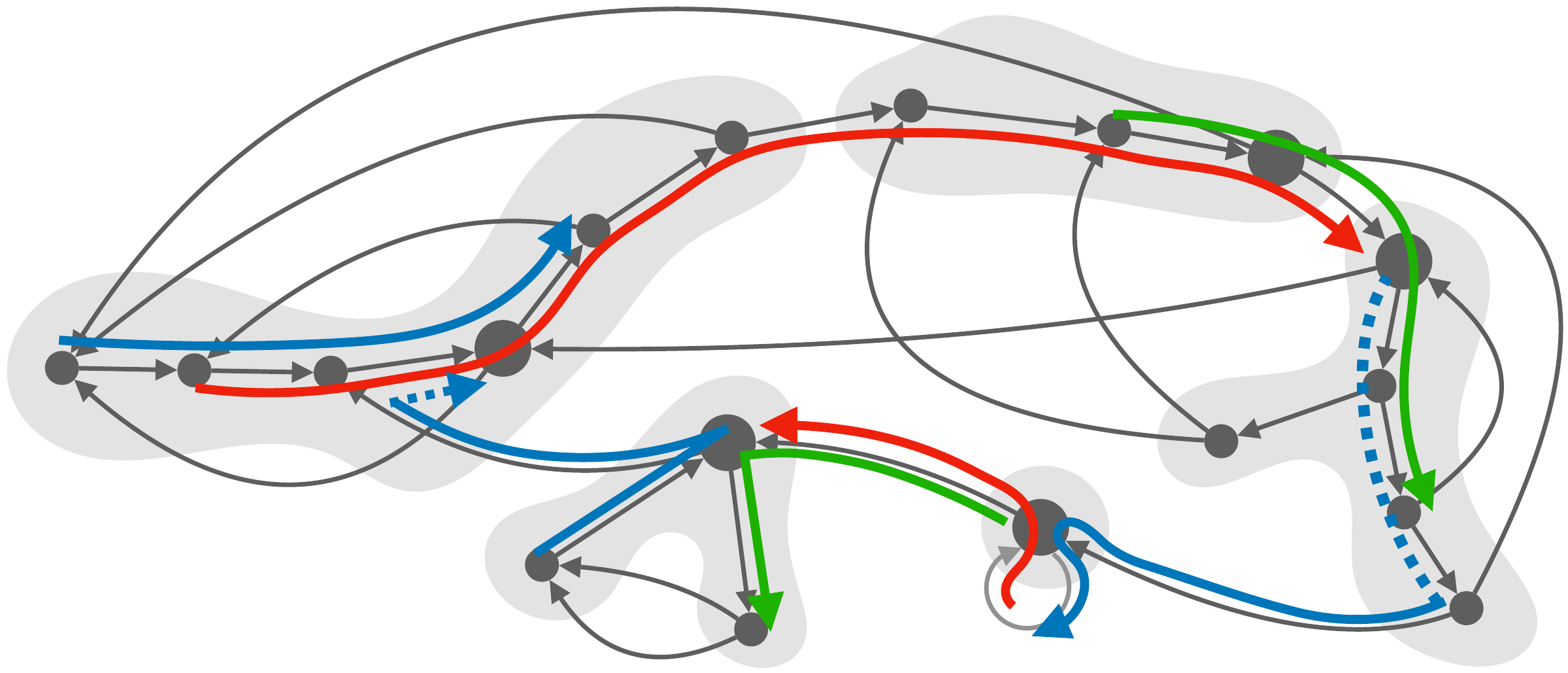}
    \caption{\small The maximal omnitigs obtained from maximal macrotigs (univocal extensions are dotted). All other maximal omnitigs are univocal extensions of the bivalent arcs not appearing in macrotigs.\label{fig:macronodes-D}}
    \end{subfigure}
}
    \caption{\small A concrete example of the main notions of this paper. In \Cref{fig:macronodes-B,fig:macronodes-C,fig:macronodes-D} walks have different colors for visual distinguishability. 
    \label{fig:macronodes}}
\end{figure}



\section{Basics}
\label{sec:preliminaries}

In this paper,
a \emph{graph} is a tuple $G=(V,E)$, where
$V$ is a finite set of \emph{nodes},
$E$ is a finite multi-set of ordered pair of nodes called \emph{arcs}.
Parallel arcs and self-loops are allowed.
For an arc $e \in E(G)$, we denote $G \smallsetminus e = (V , E \smallsetminus \{e\})$. The \emph{reverse graph} $G^R$ of $G$ is obtained by reversing the direction of every arc. In the rest of this paper, we assume a fixed strongly connected graph $G=(V,E)$, with $|V| = n$ and $|E| = m \geq n$.

A \emph{walk} in $G$ is a sequence $W = (v_0,e_1,v_1,e_2, \dots, v_{\ell-1},e_\ell,v_\ell)$, $\ell\geq 0$,
where $v_0,v_1,\dots,v_\ell \in V$, and
each $e_i$ is an arc from $v_{i-1}$ to $v_i$. Sometimes we drop the nodes $v_0,\dots,v_\ell$ of $W$, and write $W$ more compactly as $e_1\dots e_\ell$. If an arc $e$ appears in $W$, we write $e \in W$. We say that $W$ goes \emph{from $t(W) = v_0$ to $h(W) = v_\ell$}, has \emph{length} $\ell$, contains $v_1,\dots,v_{\ell-1}$ as \emph{internal nodes}, \emph{starts with $e_1$}, \emph{ends with $e_{\ell}$}, and contains $e_2,\dots,e_{\ell - 1}$ as \emph{internal arcs}. A walk $W$ is called \emph{empty} if it has length zero, and \emph{non-empty} otherwise.
There exists exactly one empty walk $\epsilon_v = (v)$
for every node $v \in V$,
and $t(\epsilon_v) = h(\epsilon_v) = v$. A walk $W$ is called \emph{closed} if it is non-empty and $t(W) = h(W)$, otherwise it is \emph{open}. The concatenation of walks $W$ and $W'$ (with $h(W)=t(W')$) is denoted $WW'$.

A walk $W = (v_0, e_1, v_1, \ldots, e_{\ell}, v_\ell)$ is called a \emph{path} when the nodes $v_0, v_1, \ldots, v_{\ell}$ are all distinct,
with the exception that $v_{\ell} = v_0$ is allowed (in which case we have either a closed or an empty path).
To simplify notation, we may denote a walk $W = (v_0, e_1, v_1, \ldots, e_{\ell}, v_\ell)$ as a sequence of arcs, i.e.~$W = e_1 \ldots e_{\ell}$. Subwalks of open walks are defined in the standard manner. For a closed walk $W = e_0\dots e_{\ell-1}$, we say that $W' = e'_0\dots e'_{j}$ is a subwalk of $W$ if there exists $i \in \{0,\dots,\ell-1\}$ such that for every $k \in \{0,\dots,j\}$ it holds that $e'_{k} = e_{(i + k) \bmod \ell}$.

%

A \emph{closed} arc-covering walk exists if and only if the graph is strongly connected. We are interested in the (safe) walks that are subwalks of all \emph{closed} arc-covering walks, characterized in~\cite{tomescu2017safe}.

\begin{theorem}[\cite{tomescu2017safe}]
\label{thm:tomescu2017safe}
Let $G$ be a strongly connected graph different from a closed path. Then a walk $W$ is a subwalk of all closed arc-covering walks of $G$ if and only if $W$ is an \emph{omnitig}.
\end{theorem}



Observe that $W$ is an omnitig in $G$ if and only if $W^R$ is an omnitig in $G^R$. Moreover, any subwalk of an omnitig is an omnitig. For every arc $e$, its univocal extension $U(e)$ is an omnitig. 
A walk $W$ satisfying a property $\mathcal{P}$ is \emph{right-maximal} (resp., \emph{left-maximal})
if there is no walk $We$ (resp., $eW$) satisfying $\mathcal{P}$.
A walk satisfying $\mathcal{P}$ is \emph{maximal} if it is left- and right-maximal w.r.t.~$\mathcal{P}$.

Notice that if $G$ is a closed path, then every \emph{walk} of $G$ is an omnitig. As such, it is relevant to find the maximal omnitigs of $G$ only when $G$ is different from a closed path. Thus, in the rest of this paper our strongly connected graph $G$ is considered to be different from a closed path, even when we do not mention it explicitly.

\section{Macronodes and macrotigs}
\label{sec:macronode-macrotigs}

In this section, unless otherwise stated, we assume that the input graph is \emph{compressed}, in the sense that it has no biunivocal nodes and arcs. In some algorithms we will also require that the graph has constant in- and out-degree. In \Cref{sec:prepro} we show how these properties can be guaranteed, by transforming any strongly connected graph $G$ with $m$ arcs, in time $O(m)$, into a compressed graph of constant degree and with $O(m)$ nodes and arcs. 

In a compressed graph all arcs are split, join or bivalent. Moreover, in compressed graphs, the following observation holds.


\begin{observation}
\label{obs:joinsplitarcs}
	Let $G$ be a compressed graph. Let $f$ and $g$ be a join and a split arc, respectively, in $G$. The following holds:
	\begin{itemize}
	\itemsep0em
	\item[$(i)$] if $fWg$ is a walk, then $W$ has an internal node which is a bivalent node;
	\item[$(ii)$] if $gWf$ is a walk, then $gWf$ contains a bivalent arc.
	\end{itemize}
\end{observation}

In the rest of this paper we will use the following technical lemmas (omitted proofs are in \Cref{subsec:compression}.).

\begin{restatable}{lemma}{omnitigbivnodearc}
\label{lem:omnitig-biv-node-arc}
Every maximal omnitig of a compressed graph contains both a join arc and a split arc.
Moreover, it has a bivalent arc or an internal bivalent node.
\end{restatable}

\vspace{-.3cm}
\begin{restatable}{lemma}{noftwice}
\label{lem:noftwice}
Let $e$ be a join or a split arc. No omnitig can traverse $e$ twice.
\end{restatable}

\vspace{-.3cm}
\begin{restatable}{lemma}{nobivntwice}
\label{lem:nobivntwice}
Let $u$ be a bivalent node. No omnitig contains $u$ twice as an internal node.  
\end{restatable}

\subsection{Macronodes}

We now introduce a natural partition of the nodes of a compressed graph; each class of such a partition (i.e. a \emph{macronode}) contains precisely one bivalent node. We identify each class with the unique bivalent node they contain. All other nodes belonging to the same class are those that either reach the bivalent node with a join-free path or those that are reached by the bivalent node with a split-free path (recall \Cref{fig:macronodes-example}). 

\begin{definition}[Macronode]
	\label{def:macronode}
	Let $v$ be a bivalent node of $G$. Consider the following sets:
	\begin{itemize}
		\itemsep0em
		\item[] $R^+(v) := \{ u \in V(G) : \exists~ \text{a join-free path from}~ v ~\text{to}~ u \}$;
		\item[] $R^-(v) := \{ u \in V(G) : \exists~ \text{a split-free path from}~ u ~\text{to}~ v \}.$ 
	\end{itemize} 
	The subgraph $\macronode{v}$ induced by $R^+(v)\cup R^-(v)$ is called the \emph{macronode centered in $v$}.
\end{definition}

\begin{restatable}{lemma}{macropartition}
\label{obs:macropartition}
In a compressed graph $G$, the following properties hold:
\begin{itemize}
	\itemsep0em
    \item[i)] The set $\{V(\macronode{v}) : \text{$v$ is a bivalent node of $G$}\}$ is a partition of $V(G)$. 
    \item[ii)] In a macronode $\mathcal{M}_v$, $R^+(v)$ and $R^-(v)$ induce two trees with common root $v$, but oriented in opposite directions. Except for the common root, the two trees are node disjoint, all nodes in $R^-(v)$ being join nodes and all nodes in $R^+(v)$ being split nodes.
    \item[iii)] The only arcs with endpoints in two different macronodes are bivalent arcs. 
\end{itemize}
\end{restatable}

To analyze how omnitigs can traverse a macronode and the degrees of freedom they have in choosing their directions within the macronode, we introduce the following definitions. \emph{Central-micro omnitigs} are the smallest omnitigs that cross the center of a macronode. \emph{Left- and right-micro omnitigs} start from a central-micro omnitig and proceed to the periphery of a macronode. Finally, we combine left- and right-micro omnitigs into microtigs (which are not necessarily omnitigs themselves); recall \Cref{fig:microtigs}.

\begin{definition}[Micro omnitigs, microtigs]
\label{def:omni_type}
Let $f$ be a join arc and $g$ be a split arc, such that $fg$ is an omnitig.
\begin{itemize}
\itemsep0em
\item The omnitig $fg$ is called a \emph{central-micro omnitig}.
\item An omnitig $fgW$ ($Wfg$, resp.) that does not contain a bivalent arc as an internal arc is called a \emph{right-micro omnitig} (respectively, \emph{left-micro omnitig}). 
\item A walk $W = W_1 fg W_2$, where $W_1 fg$ and $fg W_2$ are, respectively, a left-micro omnitig, and a right-micro omnitig, is called a \emph{microtig}.
\end{itemize}
\end{definition}

\begin{sloppypar}
Given a join arc $f$, we first find central micro-omnitigs (of the type $fg$) with the generic function $\mathsf{RightExtension}(G,f,W)$ from \Cref{alg:extensions}, where $W$ is a join-free path (possibly empty). This extension uses the following weak version of the \extprop (since $W$ is join-free). To build up the intuition, we also give a self-contained proof of this weaker result.
\end{sloppypar}

\begin{lemma}[Weak form of the \extprop~ (\Cref{thm:strong-ita-result})]
\label{lem:weak_ita_lem}
Let $fW$ be an omnitig in $G$, where $f$ is a join arc and $W$ is a join-free path. 
Then $fWg$ is an omnitig if and only if $g$ is the only arc with $t(g) = h(W)$ such that there exists a path from $h(g)$ to $h(f)$ in $G \smallsetminus f$.
\end{lemma}
\begin{proof}
To prove the existence of an arc $g$, which satisfies the condition, consider any closed path $Pf'$ in $G$, where $f'$ is an arbitrary sibling join arc of $f$. 
Notice that $W$ is a prefix of $Pf'$, since $fW$ is an omnitig, since otherwise one can easily find a forbidden path for the omnitig $fW$ as a subpath of $Pf'$, from the head of the very first arc of $Pf'$ that is not in $W$ to $h(f')$. Therefore, let $g$ be the the first arc of $Pf'$ after the prefix $W$, in such a way that the suffix of $Pf'$ starting from $h(g)$ is a path to $h(f)$ in $G \smallsetminus f$.

For the direct implication, assume that there is a path $P$ in $G \smallsetminus f$ from $h(g')$, where $g'$ sibling of $g$ and $g' \neq g$, to $h(f)$. Then, this forbidden path $P$ contradicts the fact that $fWg$ is an omnitig.

For the reverse implication, assume that $fWg$ is not an omnitig. Then take any forbidden path $P$ for $fWg$. Since $fW$ is an omnitig, $P$ must start with some $g'$ sibling arc of $g, g' \neq g$. Since $W$ is join-free, then $P$ must end in $h(f)$ with the last arc different from $f$. Therefore, $P$ is a path from $h(g')$ to $h(f)$ in $G \smallsetminus f$.
\end{proof}

Not only \Cref{lem:weak_ita_lem} gives us an efficient extension mechanism, but it also immediately implies the \yintprop (for clarity of reusability, we state both its symmetric variants). 

\begin{corollary}[\yintprop]
\label{cor:Y-lemma}
Let $fWg$ be an omnitig, where $f$ is a join arc, and $g$ is a split arc.
\begin{itemize}
	\itemsep0em
    \item[i)] If $W$ is a join-free path (possibly empty), then for any $g'$ a sibling split arc of $g$, the walk $fWg'$ is not an omnitig.
    \item[ii)] If $W$ is a split-free path (possibly empty), then for any $f'$ a sibling join arc of $f$, the walk $f'Wg$ is not an omnitig.
\end{itemize}
\end{corollary}

We now use the \yintprop to prove the \xintprop.

\begin{proof}[Proof of the \xintprop~(\Cref{thm:Xthm})]
For point $i)$, assume there exists an arc $g_3$, distinct from $g_1$ and $g_2$, such that $t(g_3) = v$. Consider any shortest closed path $g_3 P$ (with $P$ possibly empty), which exists by the strong connectivity of $G$. Let $f$ be the last arc of $P$.
If $f \neq f_1$ then $g_3P$ is a forbidden path for the omnitig $f_1 g_1$, since $g_3 \neq g_1$.
Otherwise, if $f = f_1$ then $g_3 P$ is a forbidden path for the omnitig $f_2 g_2$, since $g_3 \neq g_1$.
In both cases we reached a contradiction, therefore $g_1$ and $g_2$ are the only arcs in $G$ with $t(g_1) = t(g_2) = v$. 
To prove that $f_1$ and $f_2$ are the only arcs in $G$ with $h(f_1) = h(f_2) = v$ one can proceed by symmetry.
 
Point $ii)$ follows from \Cref{cor:Y-lemma} (by taking the $W$ path of its statement to be empty) and from the symmetric analogue of \Cref{cor:Y-lemma}.
\end{proof}

\begin{sloppypar}
Given an omnitig $fg$, we obtain the maximal right-micro omnitig with function $\mathsf{MaximalRightMicroOmnitig}(G,f,g)$ from \Cref{alg:extensions}. This works by extending $fg$, as much as possible, with the function $\mathsf{RightExtension}(G,f,W)$ (where initially $W = g$). This extension stops when reaching the periphery of the macronode (i.e.~a bivalent arc).
\end{sloppypar}

\begin{algorithm}[t]
    \caption{Functions $\mathsf{RightExtension}$ and $\mathsf{MaximalRightMicroOmnitig}$.}
    \label{alg:extensions}
    \Function{$\mathsf{RightExtension}(G,f,W)$}{
	    \Input{The compressed graph $G$, $fW$ omnitig with $W$ join-free.}
	    \Returns{The unique arc $e$ such that $fWe$ is an omnitig, if it exists. Otherwise, $\Nil$.}
		\BlankLine
		
		$S \gets \{ e \in E(G) \mid t(e) = h(W) \text{ and there is a path from } h(e) \text{ to } h(f) \text{ in } G \smallsetminus f \}$ \;
		
		\lIf{there is exactly one arc $e \in S$}{\Return $e$}
		
    	\BlankLine

		\Return $\Nil$ \;
	}
	\BlankLine
	\Function{$\mathsf{MaximalRightMicroOmnitig}(G,f,g)$}{
		\Input{The compressed graph $G$, $fg$ omnitig with $f$ join arc and $g$ split arc.}
		\Returns{The path $W$ such that $fgW$ is a maximal right-micro omnitig.}
		
    	\BlankLine

        $W \gets$ empty path

	    \While{\True}{
    	    \lIf{$fgW$ ends with a bivalent arc}{\Return $W$}
    
    	    $e \gets \mathsf{RightExtension}(G,f,W)$ \;
    	    \lIf{$e = \Nil$}{\Return $W$}

        	\BlankLine

            $W \gets We$ \;
	    }
    }
\end{algorithm}

\begin{lemma}
\label{lem:alg1correct}
The functions in \Cref{alg:extensions} are correct. Moreover, assuming that the graph has constant degree, we can preprocess it in time $O(m+n)$ time, so that $\mathsf{RightExtension}(G,f,W)$ runs in constant time, and $\mathsf{MaximalRightMicroOmnitig}(G,f,g)$ runs in time linear in its output size.
\end{lemma}
\begin{proof}
For $\mathsf{RightExtension}(G,f,W)$, recall \Cref{lem:weak_ita_lem} and \Cref{cor:italiano_result} and that the input graph is a compressed graph, and as such every node has constant degree.

For $\mathsf{MaximalRightMicroOmnitig}(G,f,g)$, notice that every iteration of the \emph{while} loop increases the output by one arc and takes constant time, since $\mathsf{RightExtension}(G,f,W)$ runs in $O(1)$ time.
\end{proof}

\Cref{alg:maximalbigenerator} is the procedure to obtain all maximal microtigs of a compressed graph. It first finds all central micro-omnitigs $fg$ (with $\mathsf{RightExtension}(G,f,\emptyset)$), and it extends each to the right (i.e.~forward in $G$) and to the left (i.e.~forward in $G^R$) with $\mathsf{MaximalRightMicroOmnitig}$.

\begin{algorithm}[htb]
    \caption{Function $\mathsf{AllMaximalMicrotigs}$
    \label{alg:maximalbigenerator}}
    
    \Function{$\mathsf{AllMaximalMicrotigs}(G)$}{
		\Input{The compressed graph $G$.}
		\Returns{All the maximal microtigs in $G$.}
		
    	\BlankLine
    	
		$S \gets \emptyset$ \;
		
		\ForEach{bivalent node $u$ in $G$}{
    		\ForEach{join arc $f$ with $h(f) = u$}{
        		\ForEach{split arc $g$ with $t(g) = u$}{
        		    \If{$g = \mathsf{RightExtension}(G,f,\emptyset)$\label{line:right-extension-check}}{
                	    \Comment{$fg$ is a central-micro omnitig}

                        \BlankLine
                	    
                	    \Comment{applied symmetrically for left- and right-micro omnitigs}

                	    $W_1 \gets \mathsf{MaximalRightMicroOmnitig}(G^R,g,f)$ \;
                	    $W_2 \gets \mathsf{MaximalRightMicroOmnitig}(G,f,g)$ \;
                	    
                        \BlankLine

                        \BlankLine

                	    add $W_1 fg W_2$ to $S$ \;
        		    }
    		    }
		    }
		}
		
    	\BlankLine
    	
    	\Return $S$ \;
    }
\end{algorithm}

To prove the correctness of \Cref{alg:maximalbigenerator}, we need to show some structural properties of micro-omnitigs and microtigs, as follows.

\begin{lemma}
\label{lem:atmostonemicro}
Let $fg$ be a central-micro omnitig. The following hold:
\begin{itemize}
	\itemsep0em
    \item[i)] There exists at most one maximal right-micro omnitig $fgW$, and at most one maximal left-micro omnitig $Wfg$.
    \item[ii)] There exists a unique maximal microtig containing $fg$.
\end{itemize}
\end{lemma}
\begin{proof}
We prove only the first of the two symmetric statements in $i)$. If $g$ is a bivalent arc, the claim trivially holds by definition of maximal right-micro omnitig. Otherwise, a minimal counterexample consists of two right-micro omnitigs $fgPg_1$ and $fgPg_2$ (with $P$ a join-free path possibly empty), with $g_1$ and $g_2$ distinct sibling split arcs.  
Since $gP$ is a join-free path, the fact that both $fgPg_1$ $fgPg_2$ are omnitigs contradicts the \yintprop (\Cref{cor:Y-lemma}).

For $ii)$, given $fg$, by $i)$ there exists at most one maximal left-micro omnitig $W_1fg$ and at most one maximal right-micro omnitig $fgW_2$, as such there is at most one maximal microtig $W_1 fg W_2$. 
\end{proof}

\begin{lemma}
\label{lem:arcs-in-microtigs}
Let $e$ be an arc. The following hold:
\begin{itemize}
	\itemsep0em
    \item[i)] if $e$ is not a bivalent arc, then there exists at most one maximal microtig containing~$e$.
    \item[ii)] if $e$ is a bivalent arc, there exist at most two maximal microtigs containing~$e$, of which at most one is of the form $eW_1$, and at most one is of the form $W_2e$.
\end{itemize}
\end{lemma}
\begin{proof}
By symmetry, in $i)$ we only prove the case in which $e$ is a split-free arc. Notice that by \Cref{obs:macropartition}, $h(e)$ belongs to a uniquely determined macronode $\macronode{u}$ of $G$; let $P$ be the split-free path in $G$, from $h(e)$ to $u$. Let $f$ be the last arc of $eP$ ($f = e$ if $P$ is empty). By the \xintprop (\Cref{thm:Xthm}), there exists at most one split arc $g$ with $t(g) = u = h(f)$ such that $fg$ is an omnitig; if it exists, $fg$ is a central-micro omnitig, hence by \Cref{lem:atmostonemicro}, there is at most one maximal left-micro omnitig $Wfg$. Finally, if such a maximal left-micro omnitig exists, $ePg$ is a subwalk of $Wfg$, by the \yintprop (\Cref{cor:Y-lemma}). Otherwise, a minimal counterexample consists of paths $f_1 R g$ (subpath of $ePg$) and $f_2 R g$ (subpath of $Wfg$), where $f_1 \neq f_2$ and $R$ is a split-free path, since it is subpath of the split-free path $eP$; since both $f_1 R g$ and $f_2 R g$ are omnitigs, this contradicts the \yintprop.

For $ii)$, we again prove only one of the symmetric cases. The proof is identical to the above, since by \Cref{obs:macropartition}, $h(e)$ belongs to a unique macronode $\macronode{v_1}$ of $G$. As such, $e$ belongs to at most one maximal microtig $eW_1$ in $\macronode{v_1}$. Symmetrically, $t(e)$ belongs to a uniquely determined macronode $\macronode{v_2}$ of $G$. Thus, $e$ belongs to at most one maximal microtig $W_2e$ within $\macronode{v_2}$. 
\end{proof}

\begin{theorem}[Maximal microtigs]
\label{lem:maximalbigenerator}
The maximal microtigs of any strongly connected graph $G$ with $n$ nodes, $m$ arcs, and arbitrary degree have total length $O(n)$, and can be computed in time $O(m)$, with \Cref{alg:maximalbigenerator}.
\end{theorem}
\begin{proof}
First we prove the $O(n)$ bound on the total length. As we explain in \Cref{sec:prepro} we can transform $G$ into a compressed graph $G'$ such that $G'$ has $n' \leq n$ nodes and $m' \leq m$ arcs. 

Since $G'$ has at most $n'$ macronodes (recall that macronodes partition the vertex set, \Cref{obs:macropartition}), and every macronode has at most two maximal microtigs, then number of maximal microtigs is at most $2n'$. The total length of all maximal microtigs is bounded as follows. Every internal arc of a maximal microtig is not a bivalent arc, by definition. Since every non-bivalent arc appears in at most one maximal microtig (\Cref{lem:arcs-in-microtigs}), and there are at most $n'$ non-bivalent arcs in any graph with $n'$ nodes, then the number of internal arcs in all maximal microtigs is at most $n'$. Summing up for each maximal microtigs its two non-internal arcs (i.e., its first and last arc), we obtain that the total length of all maximal microtigs is at most $2n' + n' = 3n'$, thus $O(n)$.

As mentioned, in \Cref{sec:prepro} we show how to transform $G$ into a compressed graph $G'$ with $O(m)$ arcs, $O(m)$ nodes, and constant degree. On this graph we can apply \Cref{alg:maximalbigenerator}. Since every node of the graph has constant degree, the \emph{if} check in \Cref{line:right-extension-check} runs a number of times linear in the size $O(m)$ of the graph. Checking the condition in \Cref{line:right-extension-check} takes constant time, by \Cref{lem:alg1correct}; in addition, the condition is true for every central-micro omnitig $fg$ of the graph.
The \emph{then} block computes a maximal microtig and takes linear time in its size, \Cref{lem:alg1correct}. By \Cref{lem:arcs-in-microtigs} we find every microtig in linear total time.
\end{proof}

\subsection{Macrotigs}

In this section we analyze how omnitigs go from one macronode to another. Macronodes are connected with each other by bivalent arcs (\Cref{obs:macropartition}), but merging microtigs on all possible bivalent arcs may create too complicated structures. However, this can be avoided by a simple classification of bivalent arcs: those that connect a macronode with itself (\emph{self-bivalent}) and those that connect two different macronodes (\emph{cross-bivalent}), recall \Cref{fig:macronodes}.

\begin{definition}[Self-bivalent and cross-bivalent arcs]
A bivalent arc $b$ is called a \emph{self-bivalent arc} if $U(b)$ goes from a bivalent node to itself. Otherwise it is called a \emph{cross-bivalent arc}.
\end{definition}

A macrotig is now obtained by merging those microtigs from different macronodes which overlap only on a cross-bivalent arc, see also \Cref{fig:macrotig-example}.

\begin{definition}[Macrotig]
\label{def:macrotig}
Let $W$ be any walk. 
$W$ is called a \emph{macrotig} if 
\begin{enumerate}
	\itemsep0em
    \item $W$ is an microtig, or
    \item By writing $W = W_0 b_1 W_1 b_2 \ldots b_{k-1} W_{k-1} b_k W_k$, where $b_1, \ldots, b_k$ are all the internal bivalent arcs of $W$, the following conditions hold:
    \begin{enumerate}
        \item the arcs $b_1, \ldots, b_k$ are all cross-bivalent arcs, and
        \item $W_0 b_1,\; b_1 W_1 b_2,\; \ldots,\; b_{k-1} W_{k-1} b_k,\; b_k W_k$ are all microtigs.
    \end{enumerate} 
\end{enumerate}
\end{definition}

Notice that the above definition does not explicitly forbid two different macrotigs of the form $W_0 b W_1$ and $W_0 b W_2$. However,  \Cref{lem:arcs-in-microtigs} shows that there cannot be two different microtigs $bW_1$ and $bW_2$, thus we immediately obtain:

\begin{lemma}
\label{lem:uniqueprotomni}
For any macrotig $W$ there exists a unique maximal macrotig containing $W$.
\end{lemma}
\begin{proof}
W.l.o.g.,~a minimal counterexample consists of a non-right-maximal macrotig $Wb$, such that there exist two distinct microtigs $bW_1$ and $bW_2$ (notice that $b$ is a cross-bivalent arc). By \Cref{lem:arcs-in-microtigs} applied to $b$, we obtain $bW_1 = bW_2$, a contradiction.
\end{proof}

The macrotig definition also does not forbid a cross-bivalent arc to be used twice inside a macrotig. In \Cref{lem:noetwiceinproto} below we prove that also this is not possible, using the following result.

 \begin{lemma}[\cite{DBLP:journals/talg/CairoMART19}]
\label{lem:acyclic}
For any two distinct non-sibling split arcs $g, g'$, write $g \prec g'$ if there exists an omnitig $g P g'$ where $P$ is split-free. Then, the relation $\prec$ is acyclic.
\end{lemma}

\begin{figure}[t]
	\centering
	\includegraphics[width=\textwidth]{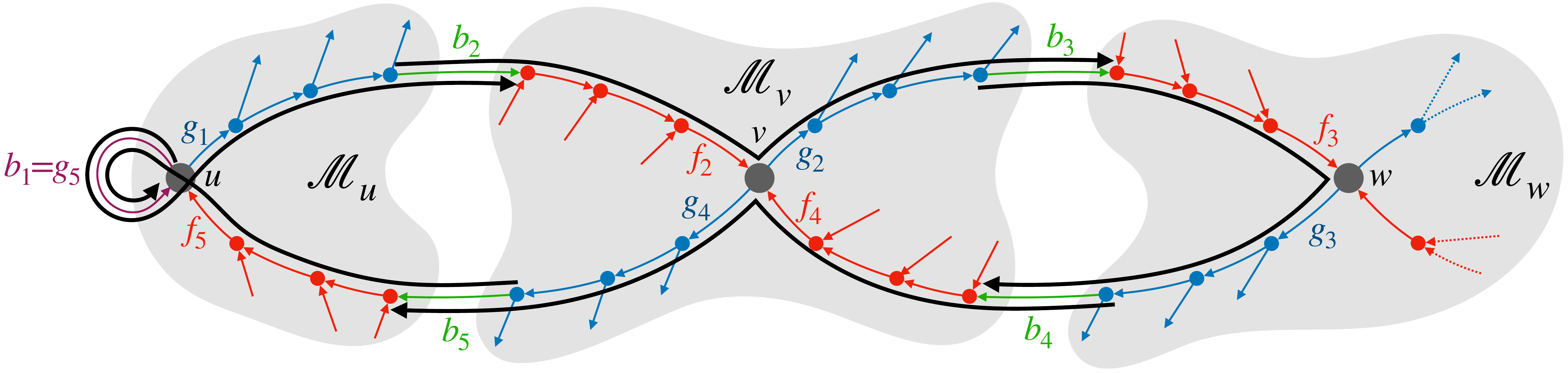}
\caption{\small Three macronodes $\mathcal{M}_{u},\mathcal{M}_{v},\mathcal{M}_{w}$ (as gray areas) with arcs color-coded as in \Cref{fig:macronodes-example}. Black walks mark their five maximal microtigs: $b_1 g_1\dots b_2$, $b_i\dots f_ig_i\dots b_{i+1}$ ($i \in \{2,3,4\}$), $b_5\dots f_5g_5$ ($g_5 = b_1$). The maximal macrotig $M$ is obtained by overlapping them on the cross-bivalent arcs $b_2,b_3,b_4,b_5$, i.e.~$M = b_1\dots b_2 \dots b_3\dots b_4\dots b_5 \dots b_1$. Notice that no arc is contained twice, with the exception of the cross-bivalent arc $b_1$, appearing as the first and last arc of $M$ (\Cref{lem:noetwiceinproto}). Bivalent nodes (e.g.~$u,v$) can appear (at most) twice in $M$ (by the \xintprop and \Cref{lem:atmostonemicro}).
\label{fig:macrotig-example}}
\end{figure}

\begin{lemma}
\label{lem:noetwiceinproto}
Let $W$ be a macrotig and let $e$ be an arc of $W$. If $e$ is self-bivalent, then $e$ appears at most twice in $W$ (as first or as last arc of $W$). Otherwise, $e$ appears only once.
\end{lemma}
\begin{proof}
If $e$ is self-bivalent, then \Cref{def:macrotig} implies that $e$ is either the first arc of $W$, the last arc of $W$, or both. Thus, $e$ appears at most twice.

Suppose now that $e$ is not self-bivalent. We first consider the case when $e$ is a split arc. We are going to prove that between any two consecutive non-self-bivalent split arcs the relation $\prec$ from \Cref{lem:acyclic} holds. Indeed, let $g$ and $g'$ be two consecutive (i.e.~closest distinct) non-self-bivalent split arcs along $W$: that is $gPg'$ subwalk of $W$, with $P$ a split-free path. Notice that $g$ and $g'$ are not sibling arcs; since otherwise, $g$ is a self-bivalent arc, by \Cref{obs:joinsplitarcs}. If $t(g')$ is not a bivalent node, then $P$ is empty. In this case, $g$ is a join-free arc, so $gg'$ is an omnitig; as such, $g \prec g'$.
Otherwise, if $t(g')$ is a bivalent node, then $gPg'$ is a left-micro omnitig and so it is an omnitig; as such, again, $g \prec g'$.

Suppose for a contradiction that $e$ is traversed twice. Since there are no internal self-bivalent arcs (as argued at the beginning of the proof), this would result in a cycle in the relation $\prec$, which contradicts \Cref{lem:acyclic}.

When $e$ is a non-self-bivalent join arc, we proceed symmetrically. First, notice that the relation defined in \Cref{lem:acyclic} is symmetric: if $f$ and $f'$ are two distinct non-sibling join arcs such that $fPf'$, with $P$ a join-free path, then $f \prec f'$. The claim above can be symmetrically adapted to hold for any two closest distinct non-self-bivalent join arcs $f$ and $f'$ within a macrotig (i.e.~corresponding to a subwalk of $W$ of the form $fPf'$, with $P$ a join-free path). Moreover, $f$ and $f'$ are not siblings; since otherwise, $f'$ is a self-bivalent arc, by \Cref{obs:joinsplitarcs}.    

Hence, by the acyclicity property of the relation $\prec$ on the reverse graph, the claim also holds for non-self-bivalent join arcs.
\end{proof}

Therefore, we can construct all \emph{maximal} macrotigs by repeatedly joining microtigs overlapping on cross-bivalent arcs, as long as possible, as in \Cref{alg:maxproto}. 

\begin{algorithm}[tb]
    \caption{Function $\mathsf{AllMaximalMacrotigs}$.}
	\label{alg:maxproto}
	
	\Function{$\mathsf{AllMaximalMacrotigs}(G)$}{
		\Input{The compressed graph $G$.}
		\Returns{All the maximal macrotigs in $G$.}
		
    	\BlankLine
    	
		$S \gets \mathsf{AllMaximalMicrotigs}(G)$ \;
		
		\While{$\exists$ $W_1 b \in S$ \And $b W_2 \in S$ with $b$ cross-bivalent arc \And non-empty $W_1, W_2$}{
		    remove $W_1 b$ and $b W_2$ from $S$ \;
		    add $W_1 b W_2$ to $S$ \;
		}
		
    	\BlankLine
    	
    	\Return $S$ \;
    }
\end{algorithm}

\begin{theorem}[Maximal macrotigs]
\label{lem:alg-extend-correct}
The maximal macrotigs of any strongly connected graph $G$ with $n$ nodes, $m$ arcs, and arbitrary degree have total length $O(n)$, and can be computed in time $O(m)$, with \Cref{alg:maxproto}.
\end{theorem}
\begin{proof}
By \Cref{lem:maximalbigenerator}, $G$ has $O(n)$ maximal microtigs, of total length $O(n)$. By \Cref{lem:noetwiceinproto}, every maximal microtig is contained in a unique maximal macrotig (and it appears only once inside such a macrotig), and the length of each maximal macrotig is at most the sum of the lengths of its maximal microtigs; thus, we have that the total length of all maximal macrotigs is at most $O(n)$.

Using \Cref{alg:maximalbigenerator}, we can get all the $O(n)$ maximal microtigs of $G$ in time $O(m)$ (\Cref{lem:maximalbigenerator}). Once we have them, we can easily implement \Cref{alg:maxproto} in $O(m)$-time. The correctness of this algorithm is guaranteed by \Cref{lem:noetwiceinproto}.
\end{proof}

\section{Maximal omnitig representation and enumeration}
\label{sec:final-algorithm}

We begin by proving the first part of \Cref{thm:macrotigs-main}. \Cref{lem:alg-extend-correct} guarantees that the total length of maximal macrotigs is $O(n)$. Thus, it remains to prove the following lemma, since \Cref{lem:uniqueprotomni} shows that any macrotig is a subwalk of a maximal macrotig.

\begin{lemma}[Maximal omnitig representation]
\label{lem:omnijoinsplitproto}
Let $W$ be a maximal omnitig. The followings hold:
\begin{itemize}
	\itemsep0em
    \item[$i)$] If $W$ contains an internal bivalent node, then $W$ is of the form $U(fW'g)$, where $f$ is the first join arc of $W$ and $g \neq f$ is the last split arc of $W$, and $W'$ is a possibly empty walk. Moreover, $fW'g$ is a macrotig.
    \item[$ii)$] Otherwise, $W$ is of the form $U(b)$, where $b$ is a bivalent arc, and $b$ does not belong to any macrotig.
\end{itemize}
\end{lemma}
\begin{proof}
To prove $i)$, let $u$ be an internal bivalent node of $W$, and let $f_u$ and $g_u$ be, respectively, the join arc and the split arc of $W$ with $h(f_u) = u = t(g_u)$; both such $f_u$ and $g_u$ exist, since $u$ is an internal node of $W$. 
Therefore, since $W$ contains at least $f_u$ and $g_u$, let $f$ and $g$ be, respectively the first join arc and the last split arc of $W$. Observe that $f$ is either $f_u$ or it appears before $f_u$ in $W$; likewise, $g$ is either $g_u$ or it appears after $g_u$ in $W$. Thus, $f$ comes before $g$, and we can write $W = W^- f W' g W^+$, where $W'$ is the subwalk of $W$, possibly empty, from $h(f)$ to $t(g)$. Therefore, by the maximality of $W$, we have $W = W^- f W' g W^+ = U(f W' g)$.

To prove that the subwalk $f W' g$ of $W$ is a macrotig, we prove by induction that \emph{any} walk of the form $fW'g$, where $f$ is a join arc and $g$ is a split arc, is a macrotig. The induction is on the length of $W'$.

\begin{description}
\item[Case 1:] $W'$ contains no internal bivalent arcs. Since $fW'g$ contains a bivalent node (\Cref{obs:joinsplitarcs}), it is of the form
$fW'g = W'_1 f'g' W'_2$, with $h(f') = t(g') = u$ bivalent node.
Notice that $W'_1 f'g' W'_2$ is an microtig and thus it is a macrotig, by definition.

\item[Case 2:] $fW'g$ contains an internal bivalent arc $b$, i.e.~$fW'g = W'_1 b W'_2$, with $W'_1, W'_2$ non empty.
By induction, $W'_1 b$ and $b W'_2$ are macrotigs and both contain a bivalent node as internal node. Suppose $b$ is a self-bivalent arc, then both $W'_1b$ and $bW'_2$ would contain the same bivalent node $u$ as internal node, contradicting \Cref{lem:nobivntwice}. Thus, $b$ is a cross-bivalent arc and $W'_1 b W'_2$ is also a macrotig, by definition.
\end{description}

For $ii)$, notice that if $W$ contains no internal bivalent node then it contains a unique bivalent arc $b$, by \Cref{lem:omnitig-biv-node-arc,obs:joinsplitarcs}. Thus, by the maximality of $W$, it holds that $W = U(b)$. It remains to prove that there is no macrotig containing $b$. 

Suppose for a contradiction that there is a maximal left-micro omnitig $M$ containing $b$. By definition, $M$ is of the form $bW_Mf_Mg_M$. Notice that $Wg_M$ is an omnitig, because $M$ is an omnitig and the arcs of $W$ before $b$ are join-free, so $Wg_M$ can have no forbidden path. This contradicts the fact that $W$ is maximal. 

Symmetrically, we have that there is no maximal right-micro omnitig containing $b$. Thus, by definition, $b$ appears in no microtig, and thus in no macrotig.
\end{proof}

\begin{remark}
The number of maximal omnitigs containing an internal bivalent node (i.e., univocal extensions of a maximal macrotig subwalk) is $O(n)$, by maximality and by the fact that the total length of maximal macrotigs is $O(n)$ (\Cref{lem:alg-extend-correct}).
\end{remark}

Next, we are going to prove the second, algorithmic, part of \Cref{thm:macrotigs-main}. By \Cref{lem:alg-extend-correct} we can compute the maximal macrotigs of $G$ in time $O(m)$. We can trivially obtain in $O(m)$ time the set $\mathcal{F}$ of arcs not appearing in the maximal macrotigs. It remains to show how to obtain the subwalks of the maximal macrotigs univocally extending to maximal omnitigs. 

We first prove an auxiliary lemma needed for the proof of the \extprop (\Cref{thm:strong-ita-result}).

\begin{lemma}
\label{lem:doesnotcontainf}
Let $fW$ be an omnitig, where $f$ is a join arc. Let $P$ be a path from $t(P) = h(W)$ to a node in $W$, such that the last arc of $P$ is not an arc of $fW$. Then no internal node of $P$ is a node of $W$.
\end{lemma}
\begin{proof}
Consider $P_W$ the longest suffix of $P$, such that no internal node of $P_W$ is a node of $W$. If $P_W = P$, the lemma trivially holds.
Let now $W = (u_0,e_1, u_1, e_2, \ldots, e_k, u_k)$.
Let $u_i = t(P_W)$ and $u_j = h(P_W)$. If $i \geq j$, then $P_W$ is a forbidden path for $fW$,a contradiction. 
Hence, assume $i < j < k$.
Let $f'WQ$ be a closed path. Consider the walk $Z = P_W e_{j+1} \ldots e_k Q $. Notice that $e_{i+1} \notin Z$ and $f \notin Z$. Thus $Z$ can transformed in a forbidden path for $fW$, from $u_i$ to $h(f)$.
\end{proof}

\begin{proof}[Proof of the \extprop  (\Cref{thm:strong-ita-result})]
As seen in \Cref{lem:weak_ita_lem}, at least one $g$ exists which satisfies the condition.
Assume $g$ is a split arc, otherwise the statement trivially holds.

First, assume that there is a $g'$ sibling split arc of $g$ and a path $P$ from $h(g)$ to $h(f)$ in $G \smallsetminus f$. We prove that there exists a forbidden path for $fWg$.
Let $P_W$ be the prefix of $P$ ending in the first occurrence of a node in $W$ (i.e., no node of $P_W$ belongs to $W$, except for $h(P_W)$). Notice that $g'P_W$ is a forbidden path for the omnitig $fWg$ (it is possible, but not necessary, that $h(P_W) = h(f)$).
    
Second, take any forbidden path $P$ for the omnitig $fWg$. We prove that there exists a $g'$ sibling split arc of $g$ and a path from $h(g)$ to $h(f)$ in $G \smallsetminus f$.
Notice that $t(P) = h(W) = t(g)$, otherwise $P$ would be a forbidden path for $fW$. As such, $P$ starts with a split arc $g' \neq g$ and, by \Cref{lem:doesnotcontainf}, $P$ does not contain $f$. Thus, the suffix of $P$ from $h(g')$ is a path in $G \smallsetminus f$ from $h(g')$ to $h(f)$.
\end{proof}

To describe the algorithm that identifies all maximal omnitigs (\Cref{alg:maximalomnitigs}), we first introduce an auxiliary procedure (\Cref{alg:isomnirext}), which uses the \extprop (\Cref{cor:italiano_result}) and \Cref{thm:strong-ita-result} to find the unique possible extension of an omnitig.

\begin{algorithm}[htb]
\SetKw{Output}{output}
	\caption{Function $\mathsf{IsOmnitigRightExtension}$}\label{alg:isomnirext}
		
	\BlankLine
	
	\Function{$\mathsf{IsOmnitigRightExtension}(G,f,g)$}{
    	\Input{%
    		The compressed graph $G$.
    		A join arc $f$ and a split arc $g$ such that there exists a walk $fWg$ where $fW$ is an omnitig.
    	}
    	\Returns{%
    		Whether $fWg$ is also an omnitig.
    	}
    	
    	\BlankLine

		$S \gets \{ g' \in E(G) \mid t(g') = t(g) \text{ and there is a path from } h(g') \text{ to } h(f) \text{ in } G \smallsetminus f \}$ \;
		
    	\BlankLine

		\Return \True if $S = \{ g \}$ and \False otherwise
	}

	\BlankLine
\end{algorithm}

\begin{corollary}
\label{cor:isomnirext}
\Cref{alg:isomnirext} is correct. Moreover, assuming that the graph has constant degree, we can preprocess it in time $O(m+n)$ time, so that \Cref{alg:isomnirext} runs in constant time.
\end{corollary}

Maximal omnitigs are identified with a two-pointer scan of maximal macrotigs (\Cref{alg:maximalomnitigs}): a left pointer always on a join arc $f$ and a right pointer always on a split arc $g$, recall \Cref{fig:macronode-algorithm}. For the sake of completeness, we write \Cref{alg:maximalomnitigs} so that it also outputs the maximal omnitigs. In \Cref{subsec:constant-degree-enumeration} we explain what changes  are needed when the graph does not have constant degree.

\begin{algorithm}[htb]
\SetKw{Output}{output}
	\caption{Computing all maximal omnitigs\label{alg:maximalomnitigs}}
		
	\Input{The compressed graph $G$.}
	\Outputs{All maximal omnitigs of $G$.}	

	\BlankLine
	
	$ B \gets \{ b \text{ bivalent arc} \mid b \text{ does not occur in any } W \in \mathsf{AllMaximalMacrotigs}(G) \} $ \;
	\lForEach{$b \in B$\label{line:maxomniwithb}}{\Output $U(b)$ }

    \BlankLine
	
	\ForEach{$f^* X g^* \in \mathsf{AllMaximalMacrotigs}(G)$}{
	    
	    \BlankLine
	    
        \Comment{With the notation $X[f..g]$, we refer to the subwalk of $f^*Wg^*$ starting with the occurrence of $f$ in $f^* X$ (unique by \Cref{lem:noetwiceinproto}) and ending with the occurrence of $g$ in $X g^*$ (unique by \Cref{lem:noetwiceinproto}).}
        
        \BlankLine
	
		$f \gets f^*$, $g \gets \Nil$, $g' \gets $ first split arc in $X g^*$ \;
        
        \BlankLine
        
		\While{$g' \neq \Nil$ \label{line:extwhile}}{
		
	    	\While{%
	    	    $g' \neq \Nil$ \And $\mathsf{IsOmnitigRightExtension}(f, g')$ \label{line:firstintwhile}
    	    }{
    		    \Comment{Grow $X[f .. g]$ to the right as long as possible}
    		    $g \gets g'$ \;

            	\BlankLine

    		    $g' \gets $ next split arc in $X g^*$ after $g$ \;
    		}
    		
        	\BlankLine

		    \Comment{$X[f .. g]$ cannot be grown to the right anymore}
    		\Output $U(X[f .. g])$ \;
    
        	\BlankLine

	    	\While{%
	    	    $g' \neq \Nil$ \And \Not $\mathsf{IsOmnitigRightExtension}(f, g')$ \label{line:scndintwhile}
    	    }{
		        \Comment{Shrink $X[f .. g]$ from the left until it can be grown to the right again}
    		    $i \gets $ index of next join arc in $f^* X$ after $f$
    		}
		}
	}

	\BlankLine

\end{algorithm}

\begin{lemma}[Maximal omnitig enumeration]
\label{lem:alglineartime}
\Cref{alg:maximalomnitigs} is correct and, if the compressed graph has constant degree, it runs in time linear in the total size of the graph and of its output.
\end{lemma}
\begin{proof}
\Cref{alg:maxproto} returns every maximal macrotig in $O(m)$ time, by \Cref{lem:alg-extend-correct}.

By \Cref{lem:omnijoinsplitproto}, any maximal omnitig $W$ is either of the form $U(fW'g)$ (where $fW'g$ is a macrotig, and thus also a subwalk of a maximal macrotig, by \Cref{lem:uniqueprotomni}), or of the form $W = U(b)$, where $b$ is a bivalent arc not appearing in any macrotig. 

In the latter case, such omnitigs are outputted in Line~\ref{line:maxomniwithb}. 
In the former case, it remains to prove that the external \emph{while} cycle, in \Cref{line:extwhile}, outputs all the maximal omnitigs of the form $U(fW'g)$ where $fW'g$ is contained in a maximal macrotig $f^* X g^*$.

At the beginning of the first iteration, $W = U(X[f .. g'])$ is left-maximal since $f = f^*$.
The first internal \emph{while} cycle, in \Cref{line:firstintwhile}, ensures that $W = U(X[f .. g])$ is also right-maximal, at which point it is printed in output.
Then, the second internal \emph{while} cycle, in \Cref{line:scndintwhile}, ensures that $W = U(X[f .. g'])$ is a left-maximal omnitig, and the external cycle repeats.

To prove the running time bound, observe that each iteration of the \emph{foreach} cycle takes time linear in the total size of the maximal macrotig $X$ and of its output (by \Cref{cor:isomnirext}), and that the total size of all maximal macrotigs is linear, by \Cref{lem:alg-extend-correct}.
\end{proof}

\section{Constant degree and compression}
\label{sec:prepro}

In this section, we describe three transformations of the given graph $G$ to guarantee the assumption of compression and constant degree on every node. It is immediate to see that they and their inverses can be performed in linear time.

\subsection{Constant degree}
\label{subsec:constant-degree}

The first transformation allows us to reduce to the case in which the graph has constant out-degree (see \Cref{fig:transf-degree} for an example). 

\begin{transf}
	\label{tra:const_deg}
	Given $G$, for every node $v$ with $d^+(v) > 2$, let $e_1,e_2, \ldots, e_k$ be the arcs out-going from $v$. Replace $v$ with the path $(v_1,e'_1,v_2,e'_2,\ldots,e'_{k-2},v_{k-1})$, where $v_1,\dots,v_{k-1}$ are new nodes, and $e'_1,\dots,e'_{k-2}$ are new edges. Each arc $e_i$ with $t(e_i) = v$ in $G$ now has $t(e_i) = t(e'_i) = v_i$, except for $e_k$ which has $t(e_k)  = v_{k-1}$.
\end{transf}

\begin{figure}[h!]
	\centering
	\includegraphics{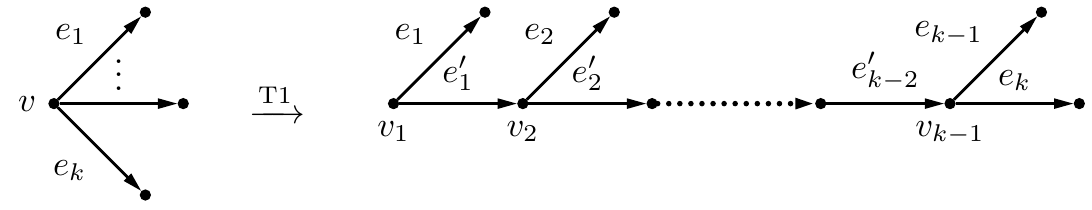}
	\caption{\small
		An example of \Cref{tra:const_deg} ($T1$) applied to the node $v$, where $e_1, \ldots, e_k \in \delta^+(v)$ are the arcs with tail equal to $v$.
	}
	\label{fig:transf-degree}
\end{figure}

By also applying the symmetric transformation, the problem on the original graph $G$ is thus reduced to a graph $G'$ with constant out- and in-degree. Notice that the number of arcs of $G'$ is still $O(m)$, where $m$ is the number of arcs of the original graph. As such, we can obtain the macrotigs of $G'$ in $O(m)$ time. The trivial strategy to obtain all maximal omnitigs of $G$ is to enumerate all maximal omnitigs of $G'$, and from these contract all the new arcs introduced by the transformation (while also removing duplicate maximal omnitigs, if necessary). However, thus may invalidate the linear-time complexity of the enumeration step, since the length of the maximal omnitigs of $G$ may be super-linear in total maximal omnitig length of $G$, see \Cref{fig:macrotigs-superlinear}. In \Cref{subsec:constant-degree-enumeration} we explain how we can easily modify the maximal omnitig enumeration step to maintain the $O(m)$ output-sensitive complexity.

\begin{figure}[h!]
	\centerline{
	\begin{subfigure}[t]{0.5\textwidth}
	\centering
	\includegraphics[scale=0.35]{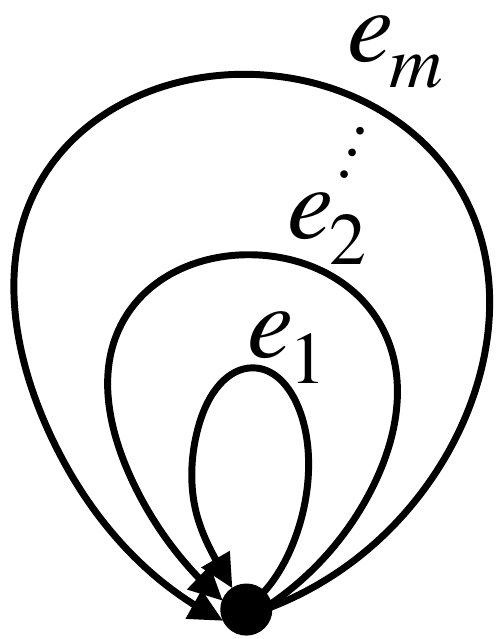}
	\end{subfigure}
	\hfill
	\begin{subfigure}[t]{0.5\textwidth}
	\centering
	\includegraphics[scale=0.35]{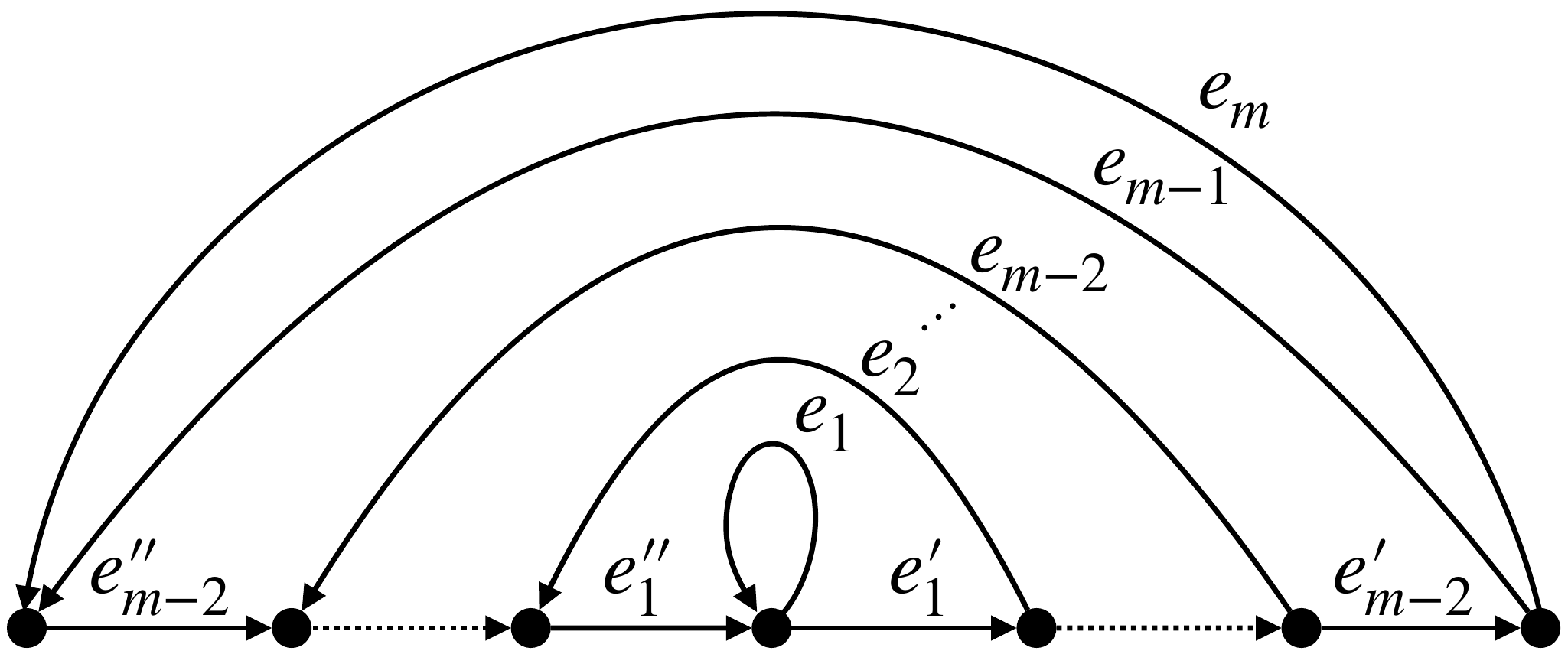}
	\end{subfigure}
    }
    \caption{\small Left: A graph $G$ made up of a single node and $m \geq 3$ self-loops $e_1,\dots,e_m$. Its $m$ maximal omnitigs are $e_1,\dots,e_m$. Right: The graph $G'$ obtained from $G$ by applying \Cref{tra:const_deg} and its symmetric transformation; the nodes of $G'$ have in-degree and out-degree at most 2. Notice that the number of arcs of $G'$ is $O(m)$. The $m$ maximal omnitigs of $G'$ are of the form $U(e_i) = e'_1\cdots e'_{i-1}e_ie''_{i-1}\cdots e''_1$ (for $i \in \{1,\dots,m\}$). Notice that their total length is $\Theta(m^2)$, thus one cannot enumerate all maximal omnitigs of $G'$ and convert these to maximal omnitigs of $G$. However, one can stop all univocal extensions of the arcs $e_i$ when reaching arcs introduced by the transformations in $G'$, see \Cref{subsec:constant-degree-enumeration}.
	\label{fig:macrotigs-superlinear}
	}
\end{figure}

To prove the correctness of \Cref{tra:const_deg}, we proceed as follows. Let $c_e(G)$ be the graph obtained from $G$ by contracting an arc $e$ (\emph{contracting $e$} means that we remove $e$ and identify its endpoints). For every walk $W$ of $G$, we denote by $c_e(W)$ the walk of $c_e(G)$, obtained from $W$ by removing every occurrence of $e$ (here we regard walks as sequences of arcs).
In the following, we regard $c_e$ as a surjective function from the family of walks of $G$ to the family of walks of $c_e(G)$.

\begin{observation}
\label{obs:contract_R-univoc}
When $e$ is a split-free or join-free arc, then $c_e$ is a bijection when restricted to the closed (arc-covering) walks, or to the open walks of $G$ whose first and last arc are different than $e$.
\end{observation}

\begin{lemma}
\label{lem:omni_contracted}
Let $e$ be a join-free arc of $G$. A walk $W'$ of $c_e(G)$ is an omnitig of $c_e(G)$ if and only if there exists an omnitig $W$ of $G$ such that $W' = c_e(W)$.
\end{lemma}
\begin{proof}
Consider the shortest walk $\underline{W}$ of $G$ such that $W' = c_e(\underline{W})$.
Notice that the first and last arc of $\underline{W}$ are different than $e$. 
Moreover, $W'$ is an omnitig of $c_e(G)$ iff $\underline{W}$ is an omnitig of $G$.
Indeed, for every circular covering $C$ of $G$ it holds that $C$ avoids $\underline{W}$
iff $c_e(C)$ avoids $W'$.
\end{proof}

\begin{corollary}
\label{cor:max_omni_contracted}
Let $e$ be a join-free arc of $G$. A walk $W'$ of $c_e(G)$ is a maximal omnitig of $c_e(G)$ if and only if there exists a maximal omnitig $W$ of $G$ such that $W' = c_e(W)$.
\end{corollary}
\begin{proof}
Let $W$ be a maximal omnitig of $G$.
Then $c_e(W)$ is an omnitig of $c_e(G)$ by \Cref{lem:omni_contracted}.
Moreover, if $W'$ was an omnitig of $c_e(G)$
strictly containing $c_e(W)$, then there would exist an omnitig $\overline{W}$ of $G$ such that $W' = c_e(\overline{W})$, by \Cref{lem:omni_contracted}.
Clearly, $\overline{W}$ would contain $W$ and contradict its maximality.
Therefore, $c_e(W)$ is a maximal omnitig of $c_e(G)$.

For the converse, let $W'$ be a maximal omnitig of $c_e(G)$.
Let $\underline{W}$ be the shortest and unique minimal walk of $G$ such that $W' = c_e(\underline{W})$.
By \Cref{lem:omni_contracted}, $\underline{W}$ is an omnitig of $G$.
Let $\overline{W}$ be any maximal omnitig of $G$ containing $\underline{W}$.
We claim that $c_e(\overline{W}) = W' = c_e(\underline{W})$, which concludes the proof.
If not, then $c_e(\overline{W})$ would strictly contain $W'$ and contradict its maximality since also $c_e(\overline{W})$ would be an omnitig of $c_e(G)$ by \Cref{lem:omni_contracted}. 
\end{proof}

\begin{restatable}{lemma}{transformationpreservesomnitigs}
\label{lem:transformation_preserves_omnitigs}
Let $G$ be a graph and let $G'$ be the graph obtained by applying \Cref{tra:const_deg} to $G$. Then a walk $W$ of $G$ is a maximal omnitig of $G$ if and only if there exists a maximal omnitig $W'$ of $G'$ such that $W$ is the string obtained from $W'$ by suppressing all the arcs introduced with the transformation.
\end{restatable}


\begin{proof}
Notice that $G$ is obtained by applying $c_e$ to each arc $e$ introduced by \Cref{tra:const_deg}, that is, to each arc of $G'$ that is not an arc of $G$. Notice that $W$ is the string obtained from $W'$ by suppressing all the arcs introduced with the transformation if and only if $W$ is obtained from $W'$ by contracting each arc $e$ introduced by \Cref{tra:const_deg}. 
Apply \Cref{cor:max_omni_contracted}.
\end{proof}

\subsection{Compression}
\label{subsec:compression}

We start by recalling the definition of compressed graph.

\begin{definition}[Compressed graph]
\label{def:compressed-graph}
    A graph $G$ is \emph{compressed} if it contains no biunivocal nodes and no biunivocal arcs.
\end{definition}

To obtain a compressed graph, we introduce two transformations. The first one removes biunivocal nodes, by replacing those paths whose internal nodes are biunivocal with a single arc from the tail of the path to its head (see \Cref{fig:transf} for an example).

\begin{transf}
	\label{tra:unitig_comp}
	Given $G$, for every longest path $P = (v_0,e_0,\dots,e_{\ell-1},v_\ell)$, $\ell \geq 2$, such that $v_1,\dots,v_{\ell-1}$ are biunivocal nodes, we remove $v_1,\dots,v_{\ell-1}$ and their incident arcs from $G$, and we add a new arc from $v_0$ to $v_{\ell}$.
\end{transf}

This transformation is widely used in the genome assembly field, and it clearly preserves the maximal omnitigs of $G$: if $P= (v_0, e_0, \ldots, e_{\ell-1}, v_\ell), \ell \geq 2$ is a path where $v_1, \ldots, v_{\ell-1}$ are biunivocal nodes, in any closed arc-covering walk of $G$, whenever $e_0$ appears it is always followed by $e_1,\dots,e_\ell$. 

The last transformation contracts the biunivocal arcs of the graph (see \Cref{fig:transf} for an example).

\begin{transf}
	\label{tra:alone_comp}
	Given $G$, we contract every biunivocal arc $e$, namely we set $t(e') = t(e)$ for every out-going arc from $h(e)$ and remove the node $h(e)$.
\end{transf}
Also this transformation preserves the maximal omnitigs of $G$ because every maximal omnitig which contains an endpoint of $e$, also contains $e$. Notice that after \Cref{tra:unitig_comp,tra:alone_comp}, the maximum in-degree and the maximum out-degree are the same as in the original graph.

In the remainder of this section we prove some lemmas stated in \Cref{sec:macronode-macrotigs}.

\begin{figure}[t!]
	\centering
	\includegraphics[scale=0.8]{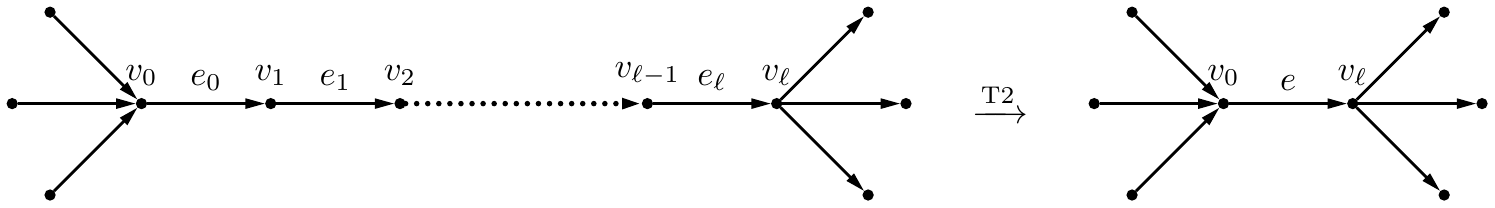}
	\includegraphics[scale=0.8]{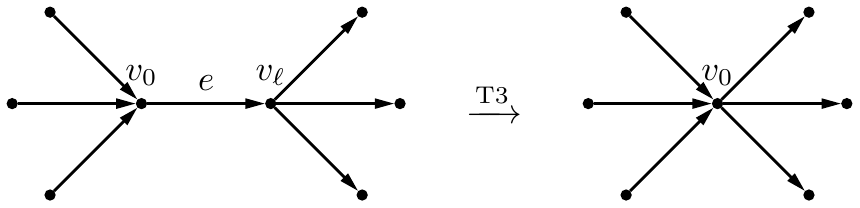}
	\caption{\small
		An example of \Cref{tra:unitig_comp} ($T2$) applied to the path $P = (v_0,e_0,\dots,e_\ell,v_\ell)$, where $v_1,\dots,v_{\ell-1}$ are biunivocal nodes and $e$ is the new arc from $v_0$ to $v_\ell$. The \Cref{tra:alone_comp} ($T3$) compresses biunivocal arcs.
	}
	\label{fig:transf}
\end{figure}

\macropartition*
\begin{proof}
For $i)$, let $u$ and $v$ be distinct bivalent nodes and suppose that there exists $x \in V(\macronode{u}) \cap V(\macronode{v})$.
W.l.o.g., assume $x$ is a join node (the case where $x$ is a split node is symmetric).
By definition, $x \in R^-(u) \cap R^-(v)$ holds. Let $P_u$ and $P_v$ be split-free paths from $x$ to $u$ and to $v$, respectively. Notice that $x$ can not be a bivalent node, since otherwise from $x$ no split-free path can start. Since the out-degree of $x$ is one, $P_u$ and $P_v$ share a prefix of length at least one, but since $u$ and $v$ are distinct bivalent nodes, $P_u$ and $P_v$ differ by at least one arc. Let $e$ be the first arc such that $e \in P_u$, but $e \notin P_v$, and let $e'$ be its sibling arc, with $e' \notin P_u$, but $e' \in P_v$. Notice that $t(e) = w$ is a join node, since it belongs to split-free paths, but it also has out-degree two, since $w = t(e) = t(e')$; hence $w$ is an internal bivalent node of split-free paths, a contradiction.

Properties $ii)$ and $iii)$ trivially follow from the definition of macronode.
\end{proof}

\omnitigbivnodearc*
\begin{proof}
Consider an omnitig $W$ composed only of split-free arcs.
Notice first that $W$ is a path.
Consider any arc $e$, with $h(e) = t(W)$ and observe that $eW$ is an omnitig, since the only out-going arcs of internal nodes of $eW$ are arcs of $eW$; thus there is no forbidden path between any two internal nodes of $eW$.
Therefore, $W$ is not a maximal omnitig.
Symmetrically, no maximal omnitig is composed only of join-free arcs.
This already implies the first claim in the statement:
any maximal omnitig $W$ contains at least one join arc $f$ and at least one split arc $g$. If $f=g$ then $W$ contains the bivalent arc $f$.
Otherwise, either $W$ contains a subwalk of the form $fW'g$
or it contains a subwalk of the form $gW'f$, where $W'$ might be an empty walk.
In the first case $W$ has an internal node which is bivalent, by \Cref{obs:joinsplitarcs}$(i)$.
In the second case $W$ contains a bivalent arc, by \Cref{obs:joinsplitarcs}$(ii)$.
\end{proof}

\noftwice*

\begin{proof}
By symmetry, we only consider the case of two sibling split arcs $g$ and $g'$. Since prefixes and suffixes of omnitigs are omnitigs, then a minimal violating omnitig would be of the form $gZg$, with $g \notin Z$. Since $G$ is strongly connected, then there exists a simple cycle $C$ of $G$ with $g' \in C$ and with $g'$ as its first arc. Notice that $g \notin C$, since $C$ is simple. Consider then the first node $u$ shared by both $C$ and $Z$, and let $e$ be the arc of $C$ with $h(e) = u$. Clearly, $e \notin Z$; in addition, $e \neq g$, since $C$ is a path. Let $C_u$ represent the prefix of $C$ ending in $u$. Therefore, $C_u$ is a forbidden path for the omnitig $gZg$, since it starts from $t(g) = t(g')$, with $g' \neq g$, and it ends in $u$ with $e \notin Z$.
\end{proof}

\nobivntwice*

\begin{proof}
Suppose for a contradiction, there exist an omnitig $W$ that contains $u$ twice as internal node. Since $u$ is an internal node of $W$, we can distinguish the case in which an omnitig contains twice a central-micro omnitig that traverses $u$, and the case in which an omnitig contains both the central-micro omnitigs that traverse $u$.
In the first case, let $fg$ be the central-micro omnitig of an omnitig $W$ that traverses $u$. Notice that $f$ is a join arc contained twice in $W$, contradicting \Cref{lem:noftwice}.
In the latter case, let $f_1g_1$ and $f_2g_2$ the two central-micro omnitigs that traverse $u$, with $f_1 \neq f_2$ and $g_1 \neq g_2$. Consider $W$ to be a minimal violating omnitig of the form $f_1g_1 \bar{W} f_2g_2$. Notice that $u \notin \bar{W}$, by minimality; hence $g_1 \bar{W} f_2$ is a forbidden path, contradicting $W$ being an omnitig.   
\end{proof}

\subsection{Maximal omnitig enumeration for non-constant degree}\label{subsec:constant-degree-enumeration}

Given the input strongly connected graph $G$ with $m$ arcs, and non-constant degree, denote by $G'$ the graph with constant in-degree and out-degree obtained by applying \Cref{tra:const_deg} and its symmetric. The trivial strategy to obtain the set of maximal omnitigs of $G$, given the set of maximal omnitigs $G'$, is to:
\begin{enumerate}
    \item Contract in the maximal omnitigs all the arcs which were introduced by \Cref{tra:const_deg}.
    \item Remove any duplicate omnitig which may occur due to this contraction (i.e., two different maximal omnitigs in $G'$ which result in the same walk in the $G$, after the contraction).
\end{enumerate}
In general, the above procedure may require more than linear time in the final output size, recall \Cref{fig:macrotigs-superlinear}.

We avoid this, as follows. Let $\mathcal{M}$ and $\mathcal{M}'$ denote the set of maximal macrotigs of $G$ and $G'$, respectively, and let $\mathcal{F}$ and $\mathcal{F}'$ denote the set of bivalent arcs not appearing in any macrotig, of $G$ and $G'$, respectively (recall \Cref{thm:macrotigs-main}). 

First, since $G'$ has $O(m)$ arcs, then also the maximal macrotigs $\mathcal{M}'$ have total length $O(m)$, and both $\mathcal{M}'$ and $\mathcal{F}'$ can be obtained in $O(m)$ time. From $\mathcal{M}'$, one can obtain $\mathcal{M}$ in time $O(m)$, by contracting the arcs introduced by the transformation. However, while contracting such arcs, we must keep track of the pair of arcs $(f,g)$ corresponding to maximal omnitigs, as follows.

We modify \Cref{alg:maximalomnitigs} to also report, for each macrotig $X'$ of $G'$ and for each maximal omnitig of the form $U(X'[f .. g])$ (in the order they were generated by the algorithm), the indexes of the arcs $f$ and $g$ in $X'$. We now contract the arcs of $X'$ by removing from $X'$ every occurrence of the arcs introduced by the transformation, and updating the indexes of $f$ and $g$ so that they still point at the first and last arc of the walk obtained from $X'[f .. g]$, after the contraction. Second, to avoid duplicates, we scan the pair of indexes of $f$ and $g$ along each macrotig, and remove any duplicated pair (if duplicates are present, they must occur consecutively, and thus they can be removed in linear time).

Second, the transformations do not introduce bivalent arcs, thus $\mathcal{F} = \mathcal{F}'$. This also implies that the arcs introduced by the transformation appear either inside macrotigs, or inside univocal extensions $U(\cdot)$. Having the set of maximal macrotigs $\mathcal{M}$ and the new arc pairs $(f,g)$ inside the maximal macrotigs in $\mathcal{M}$, it now suffices to perform the univocal extensions $U(\cdot)$ inside the original graph $G$.

\section{Acknowledgments}

We thank Sebastian Schmidt for useful comments, including the observation that the bound on the total length of all maximal macrotigs can be improved to $O(n)$ (from $O(m)$ initially), Shahbaz Khan for helpful discussions and comments, and Bastien Cazaux for discussions on the shortest superstring problem. This work was partially funded by the European Research Council (ERC) under the European Union's Horizon 2020 research and innovation programme (grant agreement No.~851093, SAFEBIO) and by the Academy of Finland (grants No.~322595, 328877).

\bibliography{bibliography}

\appendix

\section{Bioinformatics motivation}
\label{sec:bioinfo-motivation}

As mentioned in the Introduction, closed arc-covering walks were considered in \cite{TomescuMedvedev,tomescu2017safe} motivated by the genome assembly problem from Bioinformatics. We briefly review that motivation here, for the sake of completeness, and further refer the reader also to \cite{nagarajan2013sequence,DBLP:books/cu/MBCT2015}.

The genome assembly problem asks for the reconstruction of a genome string from a set $R$ of short strings (\emph{reads}) sequenced from the genome. From the read data, one usually builds a graph, and models the genome to be assembled as a certain type of walk in the graph. 

One of the most popular types of graphs is the so-called \emph{de Bruijn} graph. For a fixed integer $k$, shorter than the read length, the \emph{de Bruijn graph of order $k$} is obtained by adding a node for every $k-1$-mer (i.e.~distinct string of length $k-1$) appearing in the reads. Moreover, for every $k$-mer in the reads, one also adds an arc from the node representing its length-$(k-1)$ prefix to the node representing its length-$(k-1)$ suffix. See~\Cref{fig:db-example-1}.

\begin{figure}
    \begin{subfigure}[t]{\textwidth}
    \centering
    \includegraphics[width=0.7\textwidth]{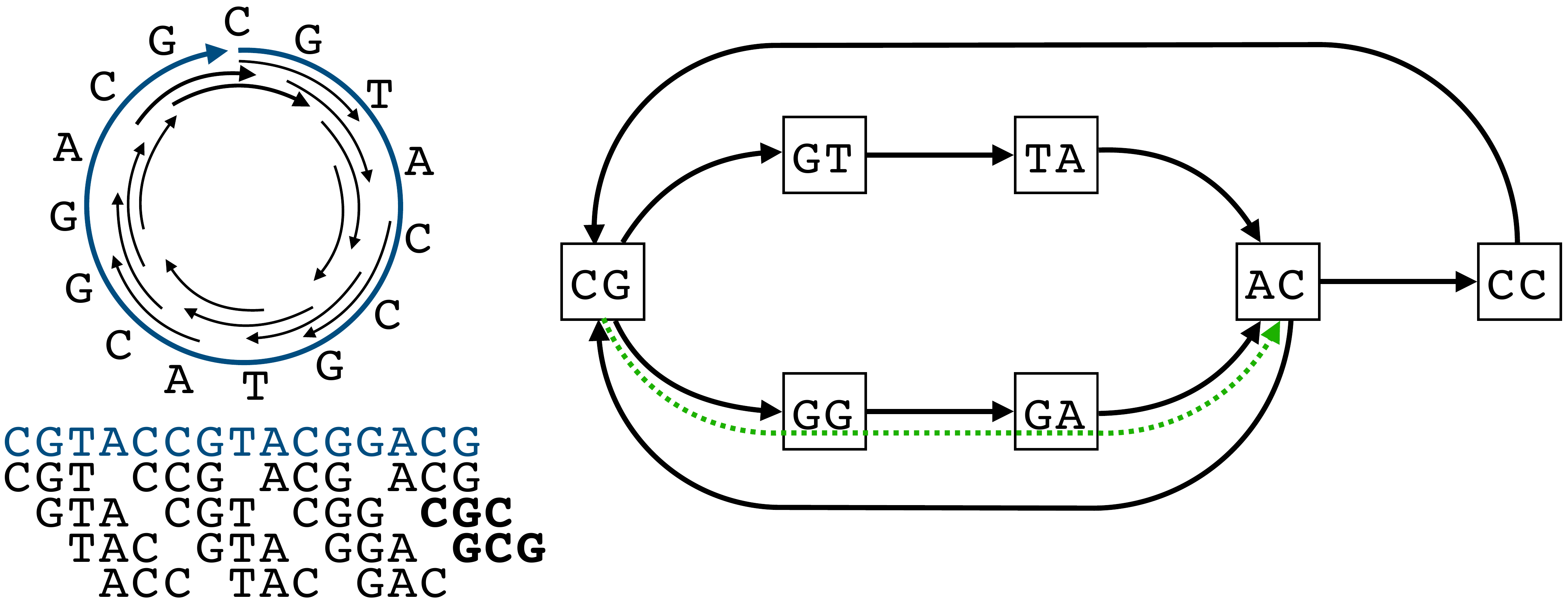}
    \caption{\small Top left: a circular string and a set of strings (\emph{reads}) of length 3 sequenced from it, drawn as black arrows. Bottom left: The same circular string now shown as linearized (in blue) and the same set of reads (in black) sequenced from it. Notice that the two reads in bold overlap the beginning and the end of the circular string. Right: the de Bruijn graph $G$ of order $k=3$ built from the set of reads. Every  distinct $(k-1)$-mer (string of length $k-1$) appearing in the reads is a node, and every $k$-mer appearing in the reads is an edge. We show as dotted-green a path appearing in all closed arc-coverings of $G$. We show the same path also as a substring of the two circular strings from~\Cref{fig:db-example-2,fig:db-example-3}.)\label{fig:db-example-1}}
    \end{subfigure}
    
	\medskip
    
    \begin{subfigure}[t]{\textwidth}
    \centering
    \includegraphics[width=0.7\textwidth]{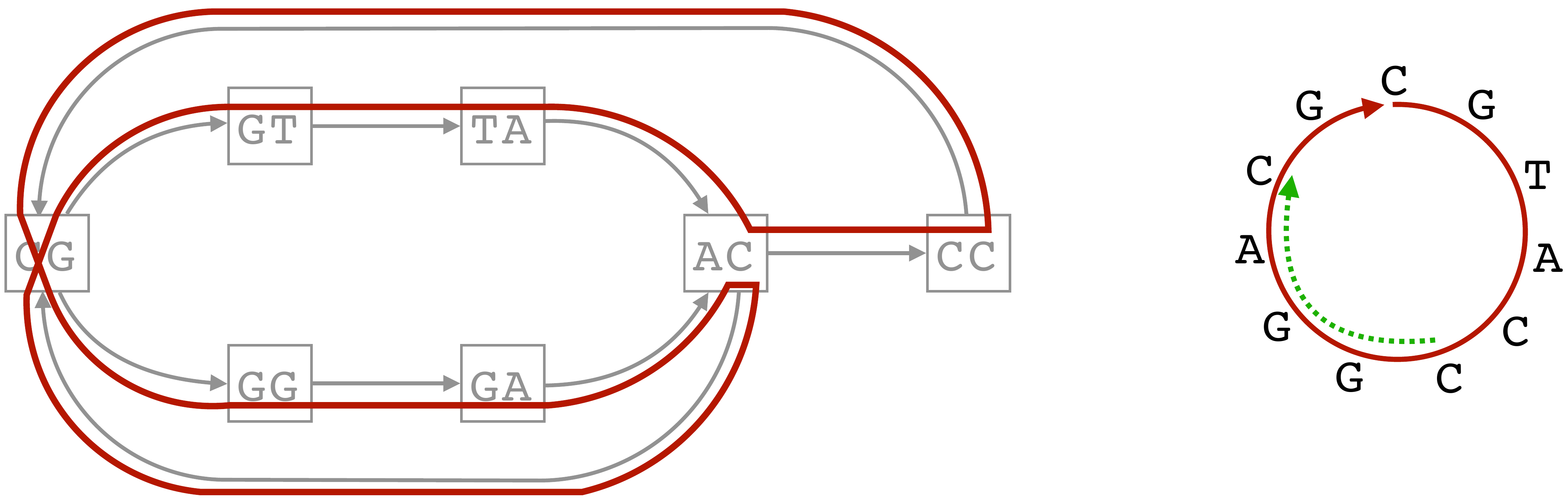}
    \caption{\small Left: In red, a closed arc-covering of $G$, which is also Eulerian. On the right, the circular string ``spelled'' by it, obtained by naturally reading the $k$-mers of its edges and merging their $(k-1)$-length overlaps. Observe that this string has the same \emph{set} of $k$-mers as the string from~\Cref{fig:db-example-1}, but is different from it.\label{fig:db-example-2}}
    \end{subfigure}
    
    \begin{subfigure}[t]{\textwidth}
    \centering
    \includegraphics[width=0.7\textwidth]{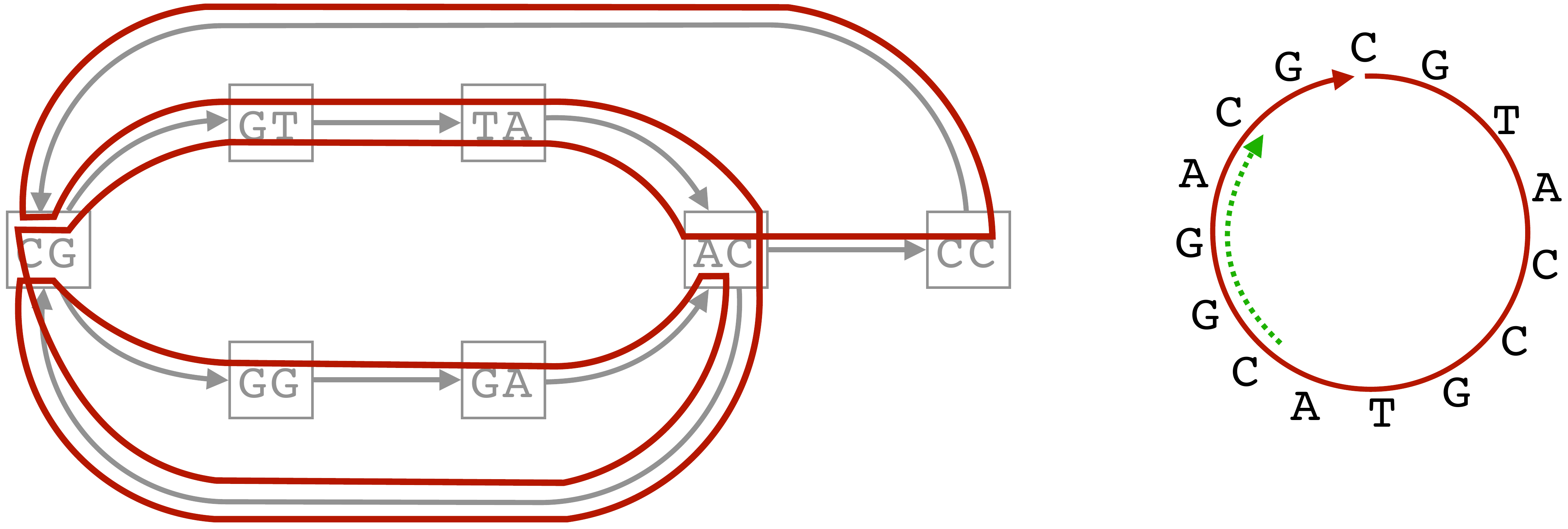}
    \caption{\small Left: In red, a closed arc-covering of $G$ spelling the same string as the circular string from~\Cref{fig:db-example-1}.\label{fig:db-example-3}}
    \end{subfigure}
    
    \caption{\small An example of a circular string, a de Bruijn graph built from a set of reads sequenced from it, two closed arc-covering walks of the graph, and a safe path appearing in all closed arc-covering walks of the graph. This safe path is biunivocal, but, as exemplified in \Cref{fig:macronodes}, other graphs can admit more complex safe walks.\label{fig:db-example}}
\end{figure}

Given such a graph, there are various ways of modeling the genome assembly solution. Assuming that the genome is circular (like in the case of most bacteria), a basic approach is to model it as a \emph{closed} Eulerian walk in the graph. Recall that, since also arcs correspond to substrings of the reads, it makes sense for the genome assembly solution to ``explain'' the arcs.\footnote{As also \cite{TomescuMedvedev,tomescu2017safe} notice, ``explaining'' arcs or nodes is mostly immaterial for de Bruijn graphs.} However, the closed Eulerian walk model is very restrictive, because of the ``exactly once'' covering requirement (in practice, the graph will not even admit a closed Eulerian walk). Another model considered in the genome assembly literature (see e.g.~\cite{nagarajan2009parametric}), and overcoming this issue, is that of a \emph{shortest} closed arc-covering walk of the graph (the Chinese Postman Problem). However, this still presents practical problems, since e.g., it collapses repeated substrings of the genome due to the minimum length requirement.

The interesting feature of both of these types of walks is that the string ``spelled'' by them (i.e., by naturally reading and merging the $k$-mers of the walk) has exactly the same \emph{set} of $k$-mers as the reads (since every $k$-mer in the reads corresponds to exactly one arc of the graph). This lead \cite{TomescuMedvedev,tomescu2017safe} to notice that closed arc-covering walks (trivially generalizing both closed Eulerian walks, and shortest closed arc-covering walk) are exactly those walks in the de Bruijn graph spelling strings with this property. Assuming that the read data is complete and error-free, then any closed arc-covering walks is a possible and valid genome assembly solution (unless also other type of data is added to the assembly problem). See~\Cref{fig:db-example}.

Looking for \emph{safe walks} with respect to closed arc-covering walks is motivated by the practical approach behind state-of-the-art genome assemblers. Such assembly programs do not report entire genome assembly solutions, because there can be an enormous number of them~\cite{kingsford2010assembly}. Instead, they report shorter strings, called \emph{contigs}, which should correspond to correct substrings of the genome. 

In most cases, and after some correction steps on the assembly graph, most genome assemblers output as contigs those strings spelling \emph{unitigs}, namely maximal biunivocal paths. Notice that unitigs appear in all closed arc-covering walks of a graph. As such,~\cite{TomescuMedvedev,tomescu2017safe} asked what are \emph{all} the safe walks (generalizing thus unitigs) for closed arc-covering walks. The answer to this question are omnitigs. The preliminary experimental results from~\cite{TomescuMedvedev,tomescu2017safe} show that under ``perfect'' conditions (complete and error-free read data), the strings spelled by omnitigs compare very favorably to unitigs, both in terms of length, and of biological content. We refer the reader to~\cite{TomescuMedvedev,tomescu2017safe} for further details.



\end{document}